\documentclass[11pt]{article}

\usepackage[utf8]{inputenc} 
\usepackage[T1]{fontenc} 
\usepackage{graphicx}
\usepackage{float} 
\usepackage{url}
\usepackage{mathtools} 
\usepackage{amsmath} 
\usepackage{amssymb} 
\usepackage{colortbl} 
\usepackage[table]{xcolor} 
\usepackage{multirow} 
\usepackage{array} 
\usepackage[margin=2.5cm]{geometry} 
\usepackage{siunitx} 
\usepackage{fancyhdr}
\usepackage[makeroom]{cancel}

\definecolor{mycitecolor}{RGB}{71, 191, 38}
\definecolor{mylinkcolor}{RGB}{40, 115, 201}
\usepackage[bookmarks=true, colorlinks, citecolor=mycitecolor,linkcolor=mylinkcolor,urlcolor=mycitecolor ]{hyperref}

\usepackage{amsthm, algorithm2e, setspace}
\usepackage{subcaption}
\usepackage{tikz}
\usetikzlibrary{patterns,angles,quotes,arrows}
\usepackage[numbers]{natbib}
\usepackage{cite, setspace}
\usepackage[inline]{enumitem}
\makeatletter
\def\thm@space@setup{\thm@preskip=10pt \thm@postskip=0pt}
\makeatother

\newtheoremstyle{mystyle}
  {}
  {}
  {}
  {}
  {\bfseries}
  {.}
  { }
  {}
\newtheorem{remark}{Remark}
\newtheorem{theorem}{Theorem}
\newtheorem{proposition}{Proposition}
\newtheorem{lemma}{Lemma}
\newtheorem{definition}{Definition}
\newtheorem{problem}{Problem}
\newtheorem{cor}{Corollary}

\DeclareMathOperator{\Image}{Im}
\DeclareMathOperator{\rank}{rank}

\DeclareMathOperator{\interior}{int}
\DeclareMathOperator{\Span}{span}
\DeclareMathOperator{\co}{co}
\DeclareMathOperator{\diag}{diag}

\title{Losing Control of your Network? Try Resilience Theory}

\author{Jean-Baptiste Bouvier\thanks{Postdoctoral Researcher in Aerospace Engineering, University of Illinois Urbana-Champaign, USA.}, Sai Pushpak Nandanoori\thanks{Staff Research Engineer, Pacific Northwest National Laboratory, Richland WA, USA.}, Melkior Ornik\thanks{Assistant Professor in Aerospace Engineering and Coordinated Science Laboratory, University of Illinois Urbana-Champaign, USA.} }

\begin{document}

\setlength{\textfloatsep}{0pt}
\setlength{\textfloatsep}{20pt plus 2pt minus 4pt}
\setlength{\textfloatsep}{10pt plus 2pt minus 4pt}
\setlength{\textfloatsep}{10pt plus 1pt minus 2pt}
\setlength{\dbltextfloatsep}{3pt}
\setlength{\intextsep}{5pt}
\setlength{\abovecaptionskip}{5pt}
\setlength{\belowcaptionskip}{3pt}
\setlength{\parskip}{4pt}
\setlength{\abovedisplayskip}{3pt}
\setlength{\belowdisplayskip}{3pt}
\setlength\abovedisplayshortskip{3pt}
\setlength\belowdisplayshortskip{3pt}
\setlist{nosep}

\date{}
\maketitle

\begin{abstract}
    Resilience of cyber-physical networks to unexpected failures is a critical need widely recognized across domains. For instance, power grids, telecommunication networks, transportation infrastructures and water treatment systems have all been subject to disruptive malfunctions and catastrophic cyber-attacks. Following such adverse events, we investigate scenarios where a node of a \emph{linear} network suffers a loss of control authority over some of its actuators. These actuators are not following the controller's commands and are instead producing undesirable outputs. The repercussions of such a loss of control can propagate and destabilize the whole network despite the malfunction occurring at a single node. To assess system vulnerability, we establish resilience conditions for networks with a subsystem enduring a loss of control authority over some of its actuators. Furthermore, we quantify the destabilizing impact on the overall network when such a malfunction perturbs a nonresilient subsystem. We illustrate our resilience conditions on two academic examples, on an islanded microgrid, and on the linearized IEEE 39-bus system.
\end{abstract}

\section{Introduction}

Resilience of cyber-physical networks to catastrophic events is a crucial challenge, widely recognized across government levels \citep{White_house, Europe} and research fields \citep{PNNL_resilience, ISCS}. Natural disasters, terrorist acts, and cyber-attacks all have the potential to paralyze the cyber-physical infrastructures upon which our society inconspicuously relies, such as power grids, telecommunication networks, sewage systems and transportation infrastructures \citep{security, actuator_attack, traffic}. 
Motivated by these issues, we investigate the resilience of linear networks to partial loss of control authority over their actuators. This class of malfunction, initially introduced in \citep{Melkior}, is characterized by some of the actuators producing uncontrolled and thus possibly undesirable inputs with their full capabilities \citep{IFAC}. This framework encompasses scenarios where actuators are taken over, for instance, by a cyber-attack \citep{security, actuator_attack, traffic}, and scenarios where the actuators become unresponsive or damaged, for instance, by a software bug \citep{ISS_thruster}.

Building on fault-detection and isolation theory \citep{actuators_measures} coupled with cyber-attack detection \citep{security} and state reconstruction methods \citep{actuator_attack}, we assume that the controller has real-time readings of the outputs of the malfunctioning actuators. Our objective is then to assess the network's resilient stabilizability in the face of these possibly undesirable inputs \citep{IFAC}.

Contrary to previous works \citep{IFAC, TAC, PNNL_resilience}, we consider actuators with bounded amplitude instead of $\mathcal{L}_2$ constraints for applicability purposes. This choice also prevents the direct use of work \citep{actuator_attack} which studies the observability properties of cyber-physical systems under unbounded adversarial attacks.

Previous works on resilience theory \citep{SIAM_CT, ECC_extended} quantified the degradation of the reachability capabilities of an isolated system enduring a partial loss of control authority over its actuators. When such a malfunctioning system is not isolated, but belongs instead to a network of interconnected systems, the loss of control can start a chain reaction capable of destabilizing the entire network. The main contribution of this work is to study these destabilizing repercussions.
Albeit using a different setting, works \citep{PNNL_resilience, network_cyber_attacks} also study the resilience of networks. Relying on observability and controllability, these works quantify the network's capabilities to detect a perturbed state and steer it back to its nominal value \citep{PNNL_resilience, network_cyber_attacks}. Because the approach of such papers does not model the perturbation it cannot handle a malfunctioning actuator producing undesirable inputs. Additionally, works \citep{PNNL_resilience, network_cyber_attacks} require $\mathcal{L}_2$ inputs, which is incompatible with numerous applications like a power grid where voltage and intensities  must remain in a specified range.

Traditionally, network resilience has been investigated through topological approaches \citep{sundaram2010distributed, leblanc2013resilient, como2012robust, xu2023interconnectedness, E2E_resilient_reachability} using the network graph to reach a consensus between all nodes \citep{dolev1986reaching, leblanc2013resilient, sundaram2010distributed}. In this setting, after a loss of control authority over $f$ nodes, at least $2f+1$ disjoint paths are required for two nodes to exchange reliable information \citep{sundaram2010distributed}. These works typically emphasize network architecture to the detriment of node dynamics, which are either unspecified \citep{dolev1986reaching, E2E_resilient_reachability, xu2023interconnectedness}, or restricted to a weighted average of neighbor states \citep{leblanc2013resilient, sundaram2010distributed}, whereas we focus on networks of control systems with generic linear dynamics. Our control framework is also broader than the domain specific resilience studies focusing for instance on public transportation networks \citep{xu2023interconnectedness}, Internet routing problems \citep{E2E_resilient_reachability}, or fluid transport networks \citep{como2012robust}.

In line with previous works studying actuator attacks \citep{actuator_attack}, network cyber attacks \citep{network_cyber_attacks}, distributed consensus \citep{sundaram2010distributed, leblanc2013resilient}, and power networks stability \citep{bidram2013distributed, bidram2013secondary, xie2019distributed, Sai_IEEE_39}, we choose to focus on networks with linearized dynamics. The reader will realize that even with linear dynamics, the resilience of networks involves a copious amount of technical calculations.

The contributions of this work are threefold.
\begin{enumerate}
    \item We establish an \emph{equivalence} condition to characterize resilient linear networks. This condition ensures that the network as a whole is resiliently stabilizable despite the loss of control authority over some actuators.
    
    \item We \emph{quantify the resilience} of fully-actuated networks having lost control over a nonresilient subsystem. More precisely, we calculate the maximal magnitude of undesirable inputs that the nonresilient subsystem can withstand without destabilizing the rest of the network by comparing the magnitude of perturbations due to subsystem couplings and their individual stability.
    
    \item We extend the resilience quantification to \emph{underactuated} networks losing control over a nonresilient subsystem. In this scenario, the malfunctioning subsystem prevents network stabilization but a feedback controller can maintain the network state within bounds.
\end{enumerate}

The remainder of this paper is organized as follows.
Section~\ref{sec: network prelim} introduces the network dynamics and states our problems of interest.
Building on prior resilience work \citep{ECC_extended}, Section~\ref{sec: resilient} establishes stabilizability conditions for resilient linear networks. 
Section~\ref{sec: nonresilient} quantifies the resilient stabilizability of networks where a loss of control authority impacts a nonresilient subsystem.
We illustrate our work on two academic examples, on an islanded microgrid, and on the linearized IEEE 39-bus system in Section~\ref{sec: example}. 
Finally, Appendix~\ref{apx} gathers supporting lemmata.

\textit{Notation:} We denote the integer interval from $a$ to $b$, inclusive, with $[\![a,b]\!]$. For a set $\Lambda \subseteq \mathbb{C}$, we say that $Re(\Lambda) \leq 0$ (resp. $Re(\Lambda) = 0$) if the real part of each $\lambda \in \Lambda$ verifies $Re(\lambda) \leq 0$ (resp. $Re(\lambda) = 0$).
The norm of a matrix $A$ is $\|A\| := \underset{x \neq 0}{\sup} \frac{\|Ax\|}{\|x\|} = \underset{\|x\| = 1}{\max} \|Ax\|$, its image is $\Image(A)$, and the set of its eigenvalues is $\Lambda(A)$.
If $A$ is positive definite, denoted $A \succ 0$, then its extremal eigenvalues are $\lambda_{min}^A$ and $\lambda_{max}^A$ and $A$ generates a vector norm $\|x\|_A := \sqrt{x^\top Ax}$. 
The controllability matrix of pair $(A,B)$ is $\mathcal{C}(A,B) := \big[B\, AB\, \hdots\, A^{n-1}B\big]$.
For a matrix $B \in \mathbb{R}^{n \times m}$ and a set $\mathcal{U} \subseteq \mathbb{R}^m$ we use $B\mathcal{U}$ to denote the set $B\mathcal{U} := \big\{ Bu : u \in \mathcal{U} \big\} \subseteq \mathbb{R}^n$. 
The block diagonal matrix composed of matrices $A_1, \hdots, A_n$ is denoted by $\diag(A_1, \hdots, A_n)$. The zero matrix of size $n \times m$ is denoted by $0_{n, m}$, the identity matrix of size $n$ is $I_n$, and the vector of ones is $\mathbf{1}_n \in \mathbb{R}^n$.
The convex hull of a set $\mathcal{Z}$ is denoted by $\co(\mathcal{Z})$, its dimension by $\dim(\mathcal{Z})$, its boundary by $\partial \mathcal{Z}$, its interior by $\interior(\mathcal{Z})$, and its orthogonal complement by $\mathcal{Z}^\perp$. 
The set of time functions taking value in $\mathcal{Z}$ is denoted $\mathcal{F}(\mathcal{Z}) := \big\{ f : [0,+\infty) \to \mathcal{Z}\big\}$.
The Minkowski addition of sets $\mathcal{X}$ and $\mathcal{Y}$ in $\mathbb{R}^n$ is $\mathcal{X} \oplus \mathcal{Y} := \big\{ x + y : x \in \mathcal{X}, y \in \mathcal{Y} \big\}$ and their Minkowski difference is $\mathcal{X} \ominus \mathcal{Y} := \big\{ z \in \mathbb{R}^n : \{z\} \oplus \mathcal{Y} \subseteq \mathcal{X} \big\}$. 
The operator $\Span(\cdot)$ maps a set of vectors to their linear span.

\section{Networks preliminaries}\label{sec: network prelim}

In this section we introduce the network under study and our two problems of interest.
Inspired by \citep{Kalman_network}, we consider a network of $N$ linear subsystems of dynamics
\begin{subequations}\label{eq:N network L F}
\renewcommand{\theequation}{ \theparentequation{}-\arabic{equation} }
    \begin{align}
         \dot x_1(t) &= A_1 x_1(t) + \bar{B}_1 \bar{u}_1(t) + \hspace{-2mm} \displaystyle\sum_{k \, \in\, \mathcal{N}_1} \hspace{-2mm} L_{1,k} y_k(t), & y_1(t) &= F_1 x_1(t), & x_1(0) &= x_1^0 \in \mathbb{R}^{n_1}, \label{eq:subsys 1 L F} \\
         & \quad \vdots &\vdots & &\vdots & \nonumber \\
         \dot x_N(t) &= A_N x_N(t) + \bar{B}_N \bar{u}_N(t) + \hspace{-2mm} \displaystyle\sum_{k \, \in\, \mathcal{N}_N} \hspace{-2mm} L_{N,k} y_k(t), & y_N(t) &= F_N x_N(t), & x_N(0) &= x_N^0 \in \mathbb{R}^{n_N}, \tag{\theparentequation-N}\label{eq:subsys N L F}
    \end{align}
\end{subequations}
where $x_i \in \mathbb{R}^{n_i}$, $\bar{u}_i \in \mathbb{R}^{m_i}$ and $y_i \in \mathbb{R}^{q_i}$ are respectively the state, the control input and the output of subsystem $i \in [\![1,N]\!]$. Additionally, $\mathcal{N}_i \subseteq [\![1,N]\!]$ is the set of neighbors of subsystem $i$ with $i \notin \mathcal{N}_i$, while $A_i \in \mathbb{R}^{n_i \times n_i}$, $\bar{B}_i \in \mathbb{R}^{n_i \times m_i}$, $L_{i,k} \in \mathbb{R}^{n_i \times q_i}$ and $F_i \in \mathbb{R}^{q_i \times n_i}$ are constant matrices. The set of admissible control inputs for subsystem $i$ is the unit hypercube of $\mathbb{R}^{m_i}$, i.e., $\bar{u}_i(t) \in \bar{\mathcal{U}}_i := [-1, 1]^{m_i}$. To alleviate the notations, we introduce matrices $D_{i,k} := L_{i,k}F_k$ to represent the connection between subsystems $i$ and $k$ for $i \in [\![1,N]\!]$ and $k \in \mathcal{N}_i$. Then,
\begin{subequations}\label{eq:N network D}
\renewcommand{\theequation}{ \theparentequation{}-\arabic{equation} }
    \begin{align}
         \dot x_1(t) &= A_1 x_1(t) + \bar{B}_1 \bar{u}_1(t) + \displaystyle\sum_{k \, \in\, \mathcal{N}_1} D_{1,k} x_k(t), \qquad& x_1(0) &= x_1^0 \in \mathbb{R}^{n_1}, \label{eq:subsys 1 D} \\
         & \quad \vdots & &\quad \vdots  \nonumber \\
         \dot x_N(t) &= A_N x_N(t) + \bar{B}_N \bar{u}_N(t) + \displaystyle\sum_{k \, \in\, \mathcal{N}_N} D_{N,k} x_k(t), \qquad& x_N(0) &= x_N^0 \in \mathbb{R}^{n_N}. \tag{\theparentequation-N}\label{eq:subsys N D}
    \end{align}
\end{subequations}

Let us now define our notion of finite-time \emph{component stabilizability}.

\begin{definition}\label{def: tuple stabilizability}
    Tuple $(A, \bar{B}, \bar{\mathcal{U}})$ is \emph{stabilizable} (resp. \emph{controllable}) if there exists a time $T \geq 0$ and an admissible control signal $\bar{u} \in \mathcal{F}(\bar{\mathcal{U}})$ driving the state of system $\dot x(t) = A x(t) + \bar{B} \bar{u}(t)$ from any $x^0 \in \mathbb{R}^n$ to $x(T) = 0$ (resp. to any $x_{goal} \in \mathbb{R}^n$).
\end{definition}

Following Definition~\ref{def: tuple stabilizability}, the stabilizability and controllability of tuple $(A_i, \bar{B}_i, \bar{\mathcal{U}}_i)$ characterize subsystem (\ref{eq:N network D}-$i$) as if it was isolated from its neighbors. 
These are local properties from which we will want to derive the associated overall network properties. To do so, we define the network state $X$ and control input $\bar{u}$ as
\begin{equation*}
    X := \big(x_1, x_2, \hdots, x_N\big) \in \mathbb{R}^{n_\Sigma} \quad \text{and} \quad \bar{u}(t) := \big(\bar{u}_1(t), \hdots, \bar{u}_N(t) \big) \in \bar{\mathcal{U}} := \bar{\mathcal{U}}_1 \times \hdots \times \bar{\mathcal{U}}_N \subseteq \mathbb{R}^{m_\Sigma}
\end{equation*}
with $n_\Sigma := n_1 + \hdots + n_N$ and  $m_\Sigma := m_1 + \hdots + m_N$, respectively. Network dynamics \eqref{eq:N network D} can then be written more concisely as
\begin{equation}\label{eq: network X}
    \dot X(t) = (A + D) X(t) + \bar{B} \bar{u}(t), \qquad X(0) = X_0 = \big(x_1^0, \hdots, x_N^0 \big) \in \mathbb{R}^{n_\Sigma},
\end{equation}
with the constant matrices $A := \diag(A_1, \hdots, A_N)$, $\bar{B} := \diag\big(\bar{B}_1, \hdots, \bar{B}_N\big)$ and $D := \big( D_{i,j} \big)_{(i,j) \in [\![1,N]\!]}$ with $D_{i,k} = 0$ if $k \notin \mathcal{N}_i$. Since our objective is to investigate connected networks, we assume that $D \neq 0$.

Following, for instance, a software bug or an adversarial attack \citep{actuator_attack, network_cyber_attacks}, several subsystems of network \eqref{eq:N network D} suffer a loss of control authority. We combine their dynamics into a single malfunctioning subsystem by stacking their states. Then, following Lemma~\ref{lemma: multi losses}, we reorder the subsystems so that we can consider, without loss of generality, a loss of control authority affecting only subsystem~\eqref{eq:subsys N D} and a number $p_N \in [\![1,m_N]\!]$ of its $m_N$ actuators. We then split the nominal input $\bar{u}_N$ between the remaining controlled inputs $u_N \in \mathcal{F}(\mathcal{U}_N)$, $\mathcal{U}_N = [-1,1]^{m_N - p_N}$ and the uncontrolled and possibly undesirable inputs $w_N \in \mathcal{F}(\mathcal{W}_N)$, $\mathcal{W}_N = [-1,1]^{p_N}$. We split accordingly matrix $\bar{B}_N$ into $B_N \in \mathbb{R}^{n_N \times (m_N - p_N)}$ and $C_N \in \mathbb{R}^{n_N \times p_N}$, so that the dynamics of subsystem \eqref{eq:subsys N D} become
\begin{equation}\label{eq:split system N}
    \dot x_N(t) = A_N x_N(t) + B_N u_N(t) + C_N w_N(t) + \sum_{k\, \in\, \mathcal{N}_N} D_{N,k} x_k(t), \qquad x_N(0) = x_N^0 \in \mathbb{R}^{n_N}.
\end{equation}
We adopt the resilience framework of \citep{TAC, ECC} where controller $u_N(t)$ has real-time knowledge of the undesirable inputs $w_N(t)$ thanks to sensors located on each actuator. This assumption of real-time knowledge was relaxed in \citep{JGCD} by considering a controller inflicted by a constant actuation delay. Beyond this additional layer of complexity, the resilience conditions were extremely similar to those with immediate knowledge of the perturbations, which is why we make this simplifying assumption.

Our central objective is to study how the partial loss of control authority over actuators of subsystem \eqref{eq:subsys N D} affects the \emph{stabilizability} and the \emph{controllability} of the whole network. To adapt these properties to malfunctioning system~\eqref{eq:split system N}, we first need the notion of \emph{resilient reachability} introduced in \citep{IFAC}.

\begin{definition}\label{def: resilient reachability}
    A target $x_{goal} \in \mathbb{R}^n$ is \emph{resiliently reachable} from $x(0) \in \mathbb{R}^n$ by malfunctioning system $\dot x(t) = A x(t) + B u(t) + C w(t)$ if for all $w \in \mathcal{F}(\mathcal{W})$, there exists $T \geq 0$ and $u \in \mathcal{F}(\mathcal{U})$ such that $u(t)$ only depends on $w(t)$ and the solution exists, is unique and $x(T) = x_{goal}$.
\end{definition}

\begin{definition}\label{def: tuple resilience}
    Tuple $(A, B, C, \mathcal{U}, \mathcal{W})$ is \emph{resiliently stabilizable} (resp. \emph{resilient}) if $0 \in \mathbb{R}^n$ (resp. every $x_{goal} \in \mathbb{R}^n$) is resiliently reachable from any $x(0) \in \mathbb{R}^n$ by malfunctioning system $\dot x(t) = A x(t) + B u(t) + C w(t)$.
\end{definition}

Network dynamics \eqref{eq: network X} are also impacted by the loss of control authority in subsystem \eqref{eq:subsys N D}. We define the network control input $u(t) := \big(\bar{u}_1(t), \hdots, \bar{u}_{N-1}(t), u_N(t) \big) \in \mathcal{U} := \bar{\mathcal{U}}_1 \times \hdots \times \bar{\mathcal{U}}_{N-1} \times \mathcal{U}_N \subseteq \mathbb{R}^{m_\Sigma - p_N}$. Network dynamics \eqref{eq: network X} then become
\begin{equation}\label{eq: network X split}
    \dot X(t) = (A + D) X(t) + B u(t) + C w_N(t), \qquad X(0) = X_0 = \big(x_1^0, \hdots, x_N^0 \big) \in \mathbb{R}^{n_\Sigma},
\end{equation}
with the constant matrices  $B = \diag\big(\bar{B}_1, \hdots, \bar{B}_{N-1}, B_N\big)$ and $C = \left(\begin{smallmatrix} 0_{(n_\Sigma-n_N) \times p_N} \\ C_N \end{smallmatrix}\right)$.

\begin{definition}\label{def: network resilience}
     Network~\eqref{eq: network X split} is \emph{resiliently stabilizable} (resp. \emph{resilient}) if tuple $\big( A+D, B, C, \mathcal{U}, \mathcal{W}_N\big)$ is resiliently stabilizable (resp. resilient).
\end{definition}

We are now led to the following problems of interest.

\begin{problem}\label{prob: resilient stabilizability}
    Assuming that $(A_N, B_N, C_N, \mathcal{U}_N, \mathcal{W}_N)$ is resiliently stabilizable and $(A_i, \bar{B}_i, \bar{\mathcal{U}}_i)$ is stabilizable for $i \in [\![1, N-1]\!]$, under what conditions is network \eqref{eq: network X split} resiliently stabilizable?
\end{problem}

\begin{problem}\label{prob: resilience}
    Assuming that $(A_N, B_N, C_N, \mathcal{U}_N, \mathcal{W}_N)$ is resilient and $(A_i, \bar{B}_i, \bar{\mathcal{U}}_i)$ is controllable for $i \in [\![1, N-1]\!]$, under what conditions is network \eqref{eq: network X split} resilient?
\end{problem}

Note that the resilience framework for network~\eqref{eq: network X split} allows its control input $u(t)$ to depend on undesirable input $w_N(t)$, which presupposes that all subsystems are aware of the attack.

After investigating the ideal cases of Problems~\ref{prob: resilient stabilizability} and \ref{prob: resilience} where tuple $(A_N, B_N, C_N, \mathcal{U}_N, \mathcal{W}_N)$ is resilient, we will consider the more problematic scenario where it is not resilient and study whether the other subsystems of the network are stabilizable or controllable despite the perturbations arising from the coupling with malfunctioning subsystem~\eqref{eq:split system N}.
Let $\chi(t)$ denote the combined state of all other subsystems, i.e., $\chi(t) := \big(x_1(t), \hdots, x_{N-1}(t) \big)$
Then,
\begin{equation}\label{eq: X N-1}
    \dot \chi(t) = \hat{A} \chi(t) + \hat{B} \hat{u}(t) + \hat{D} \chi(t) + D_{-,N} x_N(t),
\end{equation}
with $\chi_0 := \big(x_1^0, \hdots, x_{N-1}^0 \big)$, $\hat{A} := \diag \big(A_1, \hdots, A_{N-1}\big)$, $\hat{B} := \diag \big(\bar{B}_1, \hdots, \bar{B}_{N-1}\big)$, and $\hat{u}(t) := \big( \bar{u}_1(t), \hdots, \bar{u}_{N-1}(t) \big) \in \hat{\mathcal{U}} := \bar{\mathcal{U}}_1 \times \hdots \times \bar{\mathcal{U}}_{N-1} = [-1, 1]^{m_\Sigma - m_N}$. We also split matrix $D$ accordingly:
\begin{equation}\label{eq: notation D-}
    D = \left(\begin{array}{ccccc|c} 
    0 & D_{1,2} & \hdots & D_{1,N\text{-}2} & D_{1, N\text{-}1} & D_{1,N} \\
    D_{2,1} & 0 & & & D_{2,N\text{-}1} & D_{2,N}\\
    \vdots & & \ddots & & \vdots & \vdots \\
    D_{N\text{-}2, 1} & &  & 0 & D_{N\text{-}2, N\text{-}1} & D_{N\text{-}2,N} \\
    D_{N\text{-}1, 1} & D_{N\text{-}1,2} & \hdots & D_{N\text{-}1,N\text{-}2} & 0 & D_{N\text{-}1, N} \\ \hline
    D_{N, 1} & \hdots & & \hdots & D_{N, N\text{-}1} & 0 
    \end{array}\right) := \left(\begin{array}{cccc|c} & & & & \\  & & \hat{D} & & D_{-,N} \\ & & & & \\ \hline & & D_{N,-} & & 0 \end{array}\right).
\end{equation}
Then, the last row of $D$ without the last diagonal block is $D_{N,-}$, while the last column of $D$ without the last diagonal block is $D_{-,N}$.

\begin{definition}\label{def: subsystem stabilizability}
     System \eqref{eq: X N-1} is \emph{resiliently stabilizable} (resp. \emph{resilient}) if for every $X_0 \in \mathbb{R}^{n_\Sigma}$ (resp. and every $\chi_{goal} \in \mathbb{R}^{n_\Sigma - n_N}$) and every $w_N \in \mathcal{F}(\mathcal{W}_N)$ there exists $T \geq 0$ and $u \in \mathcal{F}(\mathcal{U})$ such that the solution to the \emph{entire network} \eqref{eq: network X split} exists, is unique and $\chi(T) = 0$ (resp. $\chi(T) = \chi_{goal}$).
\end{definition}

Definition~\ref{def: subsystem stabilizability} considers the joint stabilizability of subsystems $1$ to $N-1$ despite any undesirable input $w_N$ perturbing their combined state $\chi$ through the coupling term $D_{-,N} x_N$ in \eqref{eq: X N-1}. The resilient stabilizability of subsystem~\eqref{eq: X N-1} depends on the initial state $X_0$ of network \eqref{eq: network X split} and on the undesirable input $w_N$ perturbing state $\chi$ through the coupling term $D_{-,N} x_N$ in \eqref{eq: X N-1}.
If stabilizing state $\chi$ is impossible, the next best objective would be to maintain it around the origin.

\begin{definition}\label{def: resilient boundedness}
     System~\eqref{eq: X N-1} is \emph{resiliently bounded} if for every $X_0 \in \mathbb{R}^{n_\Sigma}$ and every $w_N \in \mathcal{F}(\mathcal{W}_N)$ there exists $b \geq 0$ and $u \in \mathcal{F}(\mathcal{U})$ such that the solution to the \emph{entire network} \eqref{eq: network X split} exists, is unique, and $\|\chi(t)\| \leq b$ for all $t \geq 0$.
\end{definition}
We can then state our third and fourth problems of interest.

\begin{problem}\label{prob: nonresilient stabilizability}
    Assuming that $(A_N, B_N, C_N, \mathcal{U}_N, \mathcal{W}_N)$ is \emph{not} resiliently stabilizable and $(A_i, \bar{B}_i, \bar{\mathcal{U}}_i)$ is stabilizable for $i \in [\![1, N-1]\!]$, under what conditions is system~\eqref{eq: X N-1} resiliently stabilizable or resiliently bounded?
\end{problem}

\begin{problem}\label{prob: nonresilient}
    Assuming that $(A_N, B_N, C_N, \mathcal{U}_N, \mathcal{W}_N)$ is \emph{not} resilient and $(A_i, \bar{B}_i, \bar{\mathcal{U}}_i)$ is controllable for $i \in [\![1, N-1]\!]$, under what conditions is system~\eqref{eq: X N-1} resilient?
\end{problem}

Note that Problems~\ref{prob: nonresilient stabilizability} and \ref{prob: nonresilient} do not try to resiliently control subsystem \eqref{eq:split system N} along with the other subsystems. Indeed, the only way to do so would rely on the coupling term $\sum D_{N,k} x_k$, which is going to $0$ as the other subsystems are getting stabilized. Therefore, malfunctioning network~\eqref{eq: network X split} is not resiliently stabilizable when tuple $(A_N, B_N, C_N, \mathcal{U}_N, \mathcal{W}_N)$ is not resiliently stabilizable.
We start by investigating Problem~\ref{prob: resilience}.

\section{Stabilizability of resilient networks}\label{sec: resilient}

In this section, we build on several background results from stabilizability and resilience theories to tackle Problems~\ref{prob: resilient stabilizability} and \ref{prob: resilience}.

\subsection{Background results}

We consider the linear time-invariant system
\begin{equation}\label{eq:initial ODE}
    \dot x(t) = Ax(t) + \bar{B} \bar{u}(t), \quad x(0) = x_0 \in \mathbb{R}^n, \quad \bar{u}(t) \in \bar{\mathcal{U}},
\end{equation}
with $A \in \mathbb{R}^{n \times n}$ and $\bar{B} \in \mathbb{R}^{n \times m}$ constant matrices and $\bar{\mathcal{U}} = [-1, 1]^m$. The controllability and stabilizability of system \eqref{eq:initial ODE} can be assessed with Corollaries~3.6 and 3.7 of Brammer \citep{Brammer}, restated here together as Theorem~\ref{thm:Brammer}.

\begin{theorem}[Brammer's conditions \citep{Brammer}]\label{thm:Brammer}
    If $\bar{\mathcal{U}} \cap \ker(\bar{B}) \neq \emptyset$ and $\interior(\co(\bar{\mathcal{U}})) \neq \emptyset$, then system \eqref{eq:initial ODE} is stabilizable (resp. controllable) if and only if $\rank\big( \mathcal{C}(A,\bar{B}) \big) = n$, $Re\big(\lambda(A)\big) \leq 0$ (resp. $=0$) and there is no real eigenvector $v$ of $A^\top$ satisfying $v^\top \bar{B} \bar{u} \leq 0$ for all $\bar{u} \in \bar{\mathcal{U}}$.
\end{theorem}

When $0 \in \interior(\bar{\mathcal{U}})$, Theorem~\ref{thm:Brammer} boils down to Sontag's stabilizability condition for systems with bounded inputs \citep{Sontag} as the eigenvector criteria can be removed.

\begin{theorem}[Sontag's condition \citep{Sontag}]\label{thm:Sontag}
    If $0 \in \interior(\bar{\mathcal{U}})$, then system \eqref{eq:initial ODE} is stabilizable (resp. controllable) if and only if $\rank\big( \mathcal{C}(A,\bar{B}) \big) = n$ and $Re\big(\lambda(A)\big) \leq 0$ (resp. $=0$).
\end{theorem}

After a loss of control authority over $p$ of the $m$ actuators of system \eqref{eq:initial ODE}, the input signal $\bar{u}$ is split between the remaining controlled inputs $u \in \mathcal{F}(\mathcal{U})$, $\mathcal{U} = [-1, 1]^{m-p}$ and the uncontrolled and possibly undesirable inputs $w \in \mathcal{F}(\mathcal{W})$, $\mathcal{W} = [-1, 1]^p$. Matrix $\bar{B}$ is accordingly split into two constant matrices $B \in \mathbb{R}^{n \times (m-p)}$ and $C \in \mathbb{R}^{n \times p}$ so that the system dynamics become
\begin{equation}\label{eq:splitted ODE}
    \dot x(t) = Ax(t) + Bu(t) + Cw(t), \quad x(0) = x_0 \in \mathbb{R}^n, \quad u(t) \in \mathcal{U}, \quad w(t) \in \mathcal{W}.
\end{equation}

Resilience conditions established in \citep{ECC_extended} use H\'ajek's approach \citep{Hajek} and hence require the following system associated to dynamics~\eqref{eq:splitted ODE}
\begin{equation}\label{eq: Hajek}
    \dot x(t) = Ax(t) + z(t), \qquad x(0) = x_0 \in \mathbb{R}^n, \qquad z(t) \in \mathcal{Z},
\end{equation}
where $\mathcal{Z} \subseteq \mathbb{R}^n$ is the Minkowski difference between the set of admissible control inputs $B\mathcal{U} := \big\{ Bu : u \in \mathcal{U}\big\}$ and the opposite of the set of undesirable inputs $C\mathcal{W} := \big\{ Cw : w \in \mathcal{W}\big\}$, i.e., 
\begin{align*}
    \mathcal{Z} &:= \big[ B\mathcal{U} \ominus (-C\mathcal{W}) \big] \cap B\mathcal{U} \\
    &= \big\{ z \in B\mathcal{U} : \{z\} \oplus (- C\mathcal{W}) \subseteq B\mathcal{U} \big\}\\
    &= \big\{ z \in B\mathcal{U} : z - Cw \in B\mathcal{U}\ \text{for all}\ w \in \mathcal{W} \big\}.
\end{align*}
Informally, $\mathcal{Z}$ represents the control available after counteracting any undesirable input. 
The first resilience condition established in \citep{ECC_extended} is as follows.

\begin{proposition}\label{prop: resilience Z}
    If $\interior(\mathcal{Z}) \neq \emptyset$, then system \eqref{eq:splitted ODE} is resiliently stabilizable (resp. resilient) if and only if $Re\big(\lambda(A)\big) \leq 0$ (resp. $= 0$).
\end{proposition}

The main issue with Proposition~\ref{prop: resilience Z} is the requirement that $\interior(\mathcal{Z}) \neq \emptyset$ in $\mathbb{R}^n$, i.e., $\mathcal{Z}$ must be of dimension $n$, which implies that matrices $B$ and $\bar{B}$ must be full rank. To remove this restrictive requirement, the work \citep{ECC_extended} relied on a matrix $Z \in \mathbb{R}^{n \times r}$ with $r := \dim(\mathcal{Z})$ such that $\Image(Z) = \Span(\mathcal{Z})$.

\begin{theorem}[Necessary and sufficient condition \citep{ECC_extended}]\label{thm: N&S resilience}
    System \eqref{eq:splitted ODE} is resiliently stabilizable (resp. resilient) if and only if $Re\big(\lambda(A)\big) \leq 0$ (resp. $= 0$), $\rank\big( \mathcal{C}(A, Z) \big) = n$ and there is no real eigenvector $v$ of $A^\top$ satisfying $v^\top z \leq 0$ for all $z \in \mathcal{Z}$.
\end{theorem}

\begin{cor}\label{cor: N&S resilient stabilizability}
    If $\dim(\mathcal{Z}) = \rank(B)$, then system \eqref{eq:splitted ODE} is resiliently stabilizable (resp. resilient) if and only if system \eqref{eq:initial ODE} is stabilizable (resp. controllable).
\end{cor}

Notice how the resilience conditions only differ from the resilient stabilizability ones by further restricting the eigenvalues of $A$. Because of the similarity between these two concepts, from now on we will only focus on resilient stabilizability.

\subsection{Stabilizability results}\label{subsec:N stabilizability}

Malfunctioning network~\eqref{eq: network X split} can only be resiliently stabilizable if the network was stabilizable before the malfunction. We then start by investigating the finite-time stabilizability of initial network~\eqref{eq: network X}.
Since $0 \in \interior(\bar{\mathcal{U}})$, a direct application of Theorem~\ref{thm:Sontag} yields the following result.

\begin{proposition}\label{prop: N stab A+D}
    Network \eqref{eq: network X} is stabilizable if and only if $\rank\big( \mathcal{C}( A+D, \bar{B}) \big) = n_\Sigma$ and $Re(\lambda(A+D)) \leq 0$.
\end{proposition}

Since Problem~\ref{prob: resilience} aims at relating the resilient stabilizability of malfunctioning network~\eqref{eq: network X split} to that of its subsystems, a preliminary step in this direction is to relate the stabilizability of initial network~\eqref{eq: network X} to that of its subsystems unlike Proposition~\ref{prop: N stab A+D}. We will then establish several sufficient conditions for stabilizability by studying the rank and eigenvalue conditions of Proposition~\ref{prop: N stab A+D}.

First, note that having $\rank\big( \mathcal{C}(A_i, \bar{B}_i) \big) = n_i$ for all $i \in [\![1,N]\!]$ does not necessarily imply $\rank\big( \mathcal{C}( A+D, \bar{B}) \big) = n_\Sigma$, even for matrices $D$ with a small norm compared to that of $A$. We need a stronger condition on matrix $D$ to ensure that the coupling between subsystems does not alter their stabilizability. We could use the Popov-Belevitch-Hautus controllability test \citep{uncontrollability} stating the equivalence between $\rank\big( \mathcal{C}( A+D, \bar{B}) \big) = n_\Sigma$ and $\rank\big( A+D - sI, \bar{B} \big) = n_\Sigma$ for all $s \in \mathbb{C}$. In turn, this condition is equivalent to verifying whether $\bar{B}x \neq 0$ for all eigenvectors $x$ of $A+D$. However, relating the eigenvectors of $A+D$ to those of $A$ is very complicated, as detailed in Corollary 7.2.6. of \citep{matrix_computations}.

Instead, we will prefer the \emph{distance to uncontrollability} of \citep{uncontrollability} defined as
\begin{equation*}
    \mu(A,\bar{B}) := \min\big\{ \| \Delta A, \Delta \bar{B}\| : (A+ \Delta A, \bar{B}+\Delta \bar{B})\ \text{is uncontrollable}\big\} = \min\big\{ \sigma_n \big( A-sI, \bar{B}\big) : s \in \mathbb{C} \big\}.
\end{equation*}
Since $\bar{B}$ is not affected by the coupling $D$, we do not need the perturbation $\Delta \bar{B}$ present in $\mu(A,\bar{B})$ and we define instead
\begin{equation*}
    \mu_{\bar{B}}(A) := \min\big\{ \| \Delta A\| : (A+ \Delta A, \bar{B})\ \text{is uncontrollable}\big\}.
\end{equation*}
Note that $\mu_{\bar{B}}(A) \geq \mu(A, \bar{B})$ as shown in Lemma~\ref{lemma: controllability radius}.

From \citep{real_stab_radius}, we also introduce the \emph{real stability radius} of $A$
\begin{equation*}
    r_{\mathbb{R}}(A) := \inf\big\{ \|D\| : D \in \mathbb{R}^{n \times n}\ \text{and}\ A+D\ \text{is unstable} \big\}.
\end{equation*}
To approximate $r_{\mathbb{R}}(A)$ numerous lower bounds are provided in \citep{real_stab_radius}. We will now derive sufficient stabilizability conditions for network~\eqref{eq: network X} with the following statements.

\begin{proposition}\label{prop: stabilizability conditions}
    \begin{enumerate}[label=(\alph*)]
        \item If $\rank(\bar{B}_i) = n_i$ for all $i \in [\![1,N]\!]$, then $\rank\big(\mathcal{C}(A+D,\bar{B})\big) = n_\Sigma$ for all $D \in \mathbb{R}^{n \times n}$.
        \item If there exists a matrix $F \in \mathbb{R}^{m_\Sigma \times n_\Sigma}$ such that $D = \bar{B}F$ and pairs $(A_i, \bar{B}_i)$ are controllable, then $\rank\big(\mathcal{C}(A+D,\bar{B})\big) = n_\Sigma$.
        \item If $\|D\| < \mu_{\bar{B}}(A)$, then $\rank\big(\mathcal{C}(A+D,\bar{B})\big) = n_\Sigma$.
        \item If $\|D\| < r_{\mathbb{R}}(A)$, then $Re\big(\lambda(A+D)\big) \leq 0$.
    \end{enumerate}
\end{proposition}
\begin{proof}
    \begin{enumerate}[label=(\alph*)]
        \item Assume that $\rank(\bar{B}_i) = n_i$. Because $\bar{B} = \diag(\bar{B}_1, \hdots, \bar{B}_N)$, we have $\rank(\bar{B}) = n_\Sigma$, which yields $\rank\big( \mathcal{C}( A+D, \bar{B}) \big) = \rank\big( \bar{B}, (A+D)\bar{B},\hdots \big) = n_\Sigma$.
        
        \item If $D$ can be written as state feedback $D = \bar{B}F$, then $\rank\big( \mathcal{C}(A+D, \bar{B}) \big) = \rank\big( \mathcal{C}(A, \bar{B}) \big)$ \citep{multivariable_control}. Since $A$ and $\bar{B}$ are block diagonal matrices,
        \begin{equation*}
            \rank\big( \mathcal{C}(A, \bar{B}) \big) = \sum_{i=1}^N \rank\big( \mathcal{C}(A_i, \bar{B}_i) \big) = \sum_{i=1}^N n_i = n_\Sigma,
        \end{equation*}
        where $\rank\big( \mathcal{C}(A_i, \bar{B}_i) = n_i$ comes from the controllability of pair $(A_i, \bar{B}_i)$.
        
        \item By definition of $\mu_{\bar{B}}(A)$, $\|D\| < \mu_{\bar{B}}(A)$ leads to the controllability of pair $(A+D, \bar{B})$, i.e., to $\rank\big(\mathcal{C}(A+D,\bar{B})\big) = n_\Sigma$.

        \item By definition of the stability radius $\|D\| < r_{\mathbb{R}}(A)$ leads to the stability of $A+D$, i.e., $Re(\lambda(A+D)) \leq 0$.
    \end{enumerate}
\end{proof}

Combining statements (a), (b) or (c)  of Proposition~\ref{prop: stabilizability conditions} with its statement (d) yields three different sufficient stabilizability conditions for network~\eqref{eq: network X} thanks to Proposition~\ref{prop: N stab A+D}.

Note that having $\mu_{\bar{B}}(A) > 0$ and $r_{\mathbb{R}}(A) > 0$ implicitly requires the stabilizability of all the tuples $\big(A_i, \bar{B}_i, \bar{\mathcal{U}}_i\big)$. Indeed, $\mu_{\bar{B}}(A) > 0$ requires $(A, \bar{B})$ to be controllable, i.e., each tuple $(A_i, \bar{B}_i)$ must be controllable because of the diagonal structure of $A$ and $\bar{B}$. Similarly, $r_{\mathbb{R}}(A) > 0$ requires $Re\big(\lambda(A) \big) \leq 0$, but $\lambda(A) = \lambda(A_1) \cup \hdots \cup \lambda(A_N)$ because $A = \diag(A_1, \hdots, A_N)$, hence $Re\big(\lambda(A_i)\big) \leq 0$. To sum up, $\mu_{\bar{B}}(A) > 0$ and $r_{\mathbb{R}}(A) > 0$ require $\rank(A_i, \bar{B}_i) = n_i$ and $Re\big(\lambda(A_i)\big) \leq 0$, which are exactly the stabilizability conditions of Sontag for tuple  $\big(A_i, \bar{B}_i, \bar{\mathcal{U}}_i\big)$ stated in Theorem~\ref{thm:Sontag}.

Proposition~\ref{prop: stabilizability conditions} provides several stabilizability conditions for network~\eqref{eq: network X}. We will now address Problem~\ref{prob: resilient stabilizability} by studying the network dynamics after enduring a partial loss of control authority.

\subsection{Resilient stabilizability results}\label{subsec: res stabilizability}

Inspired by the work completed before Proposition~\ref{prop: resilience Z}, we define the input sets of the network $B\mathcal{U} := \big\{ Bu : u \in \mathcal{U} \big\}$, $C\mathcal{W} := \big\{ Cw_N : w_N \in \mathcal{W}_N \big\}$ and their Minkowski difference $\mathcal{Z} := B\mathcal{U} \ominus (-C\mathcal{W}) \subseteq \mathbb{R}^{n_\Sigma}$. 
Similarly, we introduce the input sets of each subsystems $\bar{B}_i \bar{\mathcal{U}}_i := \big\{ \bar{B}_i \bar{u}_i : \bar{u}_i \in \bar{\mathcal{U}}_i \big\}$ for $i \in [\![1, N-1]\!]$, $B_N \mathcal{U}_N := \big\{ B_N u_N : u_N \in \mathcal{U}_N \big\}$, $C_N \mathcal{W}_N := \big\{ C_N w_N : w_N \in \mathcal{W}_N \big\}$ and their Minkowski difference $\mathcal{Z}_N := B_N \mathcal{U}_N \ominus (-C_N \mathcal{W}_N)$. These sets are all linked together with the following result. 

\begin{proposition}\label{prop: Z cross product}
    The resilient control set of network \eqref{eq: network X split} is the Cartesian product of the input sets of its subsystems: $\mathcal{Z} = \bar{B}_1 \bar{\mathcal{U}}_1 \times \hdots \times \bar{B}_{N-1} \bar{\mathcal{U}}_{N-1} \times \mathcal{Z}_N$.
\end{proposition}
\begin{proof}
    We prove this equality by showing both inclusions.

    Take $z = (z_1, \hdots, z_N) \in \mathcal{Z}$. We want to show that $z_i \in \bar{B}_i \bar{\mathcal{U}}_i$ for $i \in [\![1, N-1]\!]$ and that $z_N \in \mathcal{Z}_N$. Let $w_N \in \mathcal{W}_N$. Since $z \in \mathcal{Z}$, there exists $u = (\bar{u}_1, \hdots, \bar{u}_{N-1}, u_N) \in \mathcal{U} = \bar{\mathcal{U}}_1 \times \hdots \times \bar{\mathcal{U}}_{N-1} \times \mathcal{U}_N$ such that 
    \begin{equation*}
        z - Cw_N = Bu = \begin{pmatrix} \bar{B}_1 \bar{u}_1 \\ \vdots \\ \bar{B}_{N-1} \bar{u}_{N-1} \\ B_N u_N \end{pmatrix} = \begin{pmatrix} z_1 \\ \vdots \\ z_{N-1} \\ z_N - C_N w_N \end{pmatrix}
    \end{equation*}
    Then, $z_i \in \bar{B}_i \bar{\mathcal{U}}_i$ for $i \in [\![1, N-1]\!]$. Additionally, for all $w_N \in \mathcal{W}_N$ we have $z_N - C_N w_N \in B_N \mathcal{U}_N$, i.e., $z_N \in \mathcal{Z}_N$.
    
    On the other hand, let $\bar{u}_i \in \bar{\mathcal{U}}_i$ for $i \in [\![1, N-1]\!]$, $z_N \in \mathcal{Z}_N$ and define $z = \big( \bar{B}_1 \bar{u}_1, \hdots, \bar{B}_{N-1}\bar{u}_{N-1}, z_N \big)$. We want to show that $z \in \mathcal{Z}$, so we take some $w_N \in \mathcal{W}_N$. Since $z_N \in \mathcal{Z}_N$, there exists $u_N \in \mathcal{U}_N$ such that $z_N - C_N w_N = B_N u_N$. Then, 
    \begin{equation*}
        z - Cw_N = \begin{pmatrix} \bar{B}_1 \bar{u}_1 \\ \vdots \\ \bar{B}_{N-1} \bar{u}_{N-1} \\ z_N \end{pmatrix} - \begin{pmatrix} 0 \\ \vdots \\ 0 \\ C_N \end{pmatrix} w_N = \begin{pmatrix} \bar{B}_1 \bar{u}_1 \\ \vdots \\ \bar{B}_{N-1} \bar{u}_{N-1} \\ B_N u_N \end{pmatrix} \in B\mathcal{U},\quad \text{so}\ z \in \mathcal{Z}.
    \end{equation*}
\end{proof}

Let us now address Problem~\ref{prob: resilient stabilizability} by considering the case where $(A_N, B_N, C_N, \mathcal{U}_N, \mathcal{W}_N)$ is resiliently stabilizable. Using Proposition~\ref{prop: resilience Z} we derive a sufficient condition for resilient stabilizability.

\begin{proposition}\label{prop: N res stab}
     If $\rank(\bar{B}_i) = n_i$ for all $i \in [\![1, N-1]\!]$, $\interior(\mathcal{Z}_N) \neq \emptyset$ and $\|D\| < r_{\mathbb{R}}(A)$, then network \eqref{eq: network X} is resiliently stabilizable.
\end{proposition}
\begin{proof}
    Since $\rank(\bar{B}_i) = n_i$, $\interior(\bar{B}_i \bar{\mathcal{U}}_i) \neq \emptyset$, so that according to Proposition~\ref{prop: Z cross product} we have $\interior(\mathcal{Z}) \neq \emptyset$. By assumption, we have $\|D\| < r_{\mathbb{R}}(A)$, i.e., $Re(\lambda(A+D)) \leq 0$. Then, Proposition~\ref{prop: resilience Z} states that network \eqref{eq: network X} is resiliently stabilizable.
\end{proof}

Proposition~\ref{prop: N res stab} provides a straightforward resilient stabilizability condition for the network in a case that is similar to Proposition~\ref{prop: stabilizability conditions}~(a). As mentioned after Proposition~\ref{prop: resilience Z}, the condition $\interior(\mathcal{Z}_N) \neq \emptyset$ requires $\rank(B_N) = \rank(\bar{B}_N) = n_N$. 
Then, Proposition~\ref{prop: N res stab} requires all $\bar{B}_i$ to be full rank, which is very restrictive and not necessary for stabilizability. Instead, we want to use Theorem~\ref{thm: N&S resilience} to derive a less restrictive resilient stabilizability condition for the network. To use this theorem, we must first build a matrix $Z \in \mathbb{R}^{n_\Sigma \times r_\Sigma}$ with $r_\Sigma = \dim(\mathcal{Z})$ and satisfying $\Image(Z) = \Span(\mathcal{Z})$. In practice, matrix $Z$ is built by collating $r_\Sigma$ linearly independent vectors from set $\mathcal{Z}$.

\begin{proposition}\label{prop: N res stab Z}
    If $\|D\| < \min\big\{ r_{\mathbb{R}}(A), \mu_Z(A) \big\}$ and there is no real eigenvector $v$ of $(A+D)^\top$ satisfying $v^\top z \leq 0$ for all $z \in \mathcal{Z}$, then network \eqref{eq: network X split} is resiliently stabilizable.
\end{proposition}
\begin{proof}
    We apply Theorem~\ref{thm: N&S resilience} to network~\eqref{eq: network X split} and obtain that it is resiliently stabilizable if and only if $Re\big(\lambda(A+D)\big) \leq 0$, $\rank\big( \mathcal{C}(A+D, Z) \big) = n_\Sigma$ and there is no real eigenvector $v$ of $(A+D)^\top$ satisfying $v^\top z \leq 0$ for all $z \in \mathcal{Z}$. The eigenvalue and rank conditions are satisfied thanks to $\|D\| < \min\{ r_{\mathbb{R}}(A), \mu_Z(A) \}$, while the eigenvector condition is verified by assumption.
\end{proof}

As before, the fact that $(A_i, \bar{B}_i, \bar{\mathcal{U}}_i)$ are stabilizable and that $(A_N, B_N, C_N, \mathcal{U}_N, \mathcal{W}_N)$ is resiliently stabilizable, are implied by the conditions of Proposition~\ref{prop: N res stab Z}. 

When $\mathcal{Z}$ is not of full dimension, the eigenvector condition of Proposition~\ref{prop: N res stab Z} is difficult to verify. Indeed, the space $\mathcal{Z}^\perp$ is non-trivial and thus might encompass a real eigenvector of $A+D$ even if none of the eigenvectors of $A$ are part of $\mathcal{Z}^\perp$.
Intuitively, when $D$ is small, the eigenvectors of $A+D$ should be `close' to those of $A$. This intuition is formalized in Corollary 7.2.6. of \citep{matrix_computations}, but the complexity of its statement prevents the derivation of a simple condition to be verified by $A$ and $D$.
Thus, we choose to remain with Propositions~\ref{prop: N res stab} and \ref{prop: N res stab Z} as our solutions to Problem~\ref{prob: resilient stabilizability}.

\section{Stabilizability of nonresilient networks}\label{sec: nonresilient}

In this section, we address Problem~\ref{prob: nonresilient stabilizability} by studying the network-wide repercussions resulting from the partial loss of control authority in nonresilient subsystem~\eqref{eq:split system N}.

We now study the eventuality where $(A_N, B_N, C_N, \mathcal{U}_N, \mathcal{W}_N)$ is not resiliently stabilizable. Following Proposition~\ref{prop: resilience Z}, we consider the case where $-C_N \mathcal{W}_N \nsubseteq B_N \mathcal{U}_N$, i.e., subsystem~\eqref{eq:split system N} lost an actuator to which it is not resilient.
Then, there are some undesirable inputs $w_N$ that no control input $u_N$ can overcome. Such undesirable inputs $w_N$ can prevent stabilizability of subsystem $N$ as demonstrated in Lemma~6 of \citep{ECC_extended}. 

To evaluate the resilient stabilizability of network~\eqref{eq: network X split}, we need to study the worst-case scenario where $w_N$ is the most destabilizing undesirable input for subsystem~\eqref{eq:split system N}. If $A_N$ is not Hurwitz, these destabilizing inputs $w_N$ can drive the state $x_N$ to infinity. In this situation, coupling terms $D_{i,N} x_N$ impacting subsystems~(\ref{eq:N network D}-i) can become unbounded preventing to stabilize these other subsystems.
We will then focus on the case where $A_N$ is Hurwitz, so that the state $x_N$ cannot be forced to diverge by $w_N$. Then, coupling term $D_{-,N} x_N$ perturbing subsystem~\eqref{eq: X N-1} is bounded and might be counteracted if controller $\hat{B} \hat{u}$ is strong enough. 

To address Problem~\ref{prob: nonresilient stabilizability}, we will quantify the maximal degree of non-resilience of subsystem \eqref{eq:split system N} despite which subsystem~\eqref{eq: X N-1} remain resiliently stabilizable in the sense of Definition~\ref{def: subsystem stabilizability}.

We start by calculating how far can $w_N$ force state $x_N$ despite the best $u_N$ and the Hurwitzness of $A_N$.

\begin{proposition}\label{prop: steady state xN}
    If $A_N$ is Hurwitz and $-C_N \mathcal{W}_N \nsubseteq B_N \mathcal{U}_N$, then for all $t \geq 0$ the following holds:
    \begin{equation}\label{eq: x_N bound incomplete}
        \|x_N(t)\|_{P_N} \leq e^{-\alpha_N t} \left( \|x_N(0)\|_{P_N} + \int_0^t e^{\alpha_N \tau} \beta_N(\tau)\, d\tau \right) ,
    \end{equation}
    for all $P_N = P_N^\top \succ 0$ and $Q_N \succ 0$ such that $A_N^\top P_N + P_N A_N = -Q_N$ and with 
    \begin{equation*}
        \alpha_N := \frac{\lambda_{min}^{Q_N}}{2\lambda_{max}^{P_N}},\ \beta_N(t) := z_{max}^{P_N} + \|D_{N,-} \chi(t)\|_{P_N},\ z_{max}^{P_N} := \underset{w_N\, \in\, \mathcal{W}_N}{\max} \Big\{ \underset{u_N\, \in\, \mathcal{U}_N}{\min} \|C_N w_N + B_N u_N\|_{P_N} \Big\}.
    \end{equation*}
\end{proposition}
\begin{proof}
    Since $A_N$ is Hurwitz, there exist a symmetric $P_N \succ 0$ and $Q_N \succ 0$ such that $A_N^\top P_N + P_N A_N = -Q_N$ according to Lyapunov theory \citep{Kalman}. Let us consider any such pair $(P_N, Q_N)$. Then, inspired by Example~15 of \citep{Kalman}, we study the $P_N$-norm of $x_N$, i.e., $x_N^\top P_N x_N = \|x_N\|_{P_N}^2$ when state $x_N$ is following dynamics~\eqref{eq:split system N}. We obtain
    \begin{align*}
        \frac{d}{dt} \|x_N(t) \|_{P_N}^2 &= \dot x_N ^\top P_N x_N + x_N^\top P_N \dot x_N \\
        &= x_N^\top \big(A_N^\top P_N + P_N A_N \big) x_N + 2 x_N^\top P_N (B_N u_N + C_N w_N) + 2 x_N^\top P_N \sum_{i=1}^{N-1} D_{N,i} x_i.
    \end{align*}
    With the notation of \eqref{eq: notation D-}, we have $\sum_{i=1}^{N-1} D_{N,i} x_i = D_{N,-} \chi$. Since $P_N \succ 0$, the Cauchy-Schwarz inequality \citep{matrix_computations} as stated in Lemma~\ref{lemma: CS} yields 
    \begin{equation*}
        x_N^\top P_N D_{N,-} \chi \leq \|x_N\|_{P_N} \|D_{N,-} \chi\|_{P_N} \hspace{2mm} \text{and} \quad x_N^\top P_N (B_Nu_N + C_Nw_N) \leq \|x_N\|_{P_N} \|B_N u_N + C_N w_N\|_{P_N}.
    \end{equation*}
    We will demonstrate the stabilizing property of the control $u_N$ minimizing $\|B_N u_N + C_N w_N\|_{P_N}$ when $w_N$ is chosen to maximize this norm. By definition these choices of $u_N$ and $w_N$ yield $\|B_N u_N + C_N w_N\|_{P_N} \leq z_{max}^{P_N}$. Then,
    \begin{equation*}
        \frac{d}{dt} \|x_N(t) \|_{P_N}^2 \leq -x_N^\top Q_N x_N + 2 \|x_N\|_{P_N} \left( z_{max}^{P_N} + \|D_{N,-} \chi\|_{P_N} \right). 
    \end{equation*}
    Since $Q_N \succ 0$, we have $-x_N^\top Q_N x_N \leq -\lambda_{min}^{Q_N} x_N^\top x_N$ \citep{Khalil} and $\|x_N\|_{P_N}^2 \leq \lambda_{max}^{P_N} x_N^\top x_N$ leads to $-x_N^\top x_N \leq \frac{-1}{\lambda_{max}^{P_N}}\|x_N\|_{P_N}^2$. Hence, we obtain
    \begin{align*}
        \frac{d}{dt} \|x_N(t)\|_{P_N}^2 &\leq -\frac{\lambda_{min}^{Q_N}}{\lambda_{max}^{P_N}} \|x_N\|_{P_N}^2 + 2 \|x_N\|_{P_N} \left( z_{max}^{P_N} + \|D_{N,-} \chi\|_{P_N} \right) \\
        &\leq -2\alpha_N \|x_N(t)\|_{P_N}^2 + 2\beta_N(t) \|x_N(t)\|_{P_N},
    \end{align*}
    by definition of $\alpha_N$ and $\beta_N$.
    We introduce $y_N(t) := \|x_N(t)\|_{P_N}$, so that we have
    \begin{equation*}
        \frac{d}{dt} y_N^2(t) = 2y_N(t) \dot y_N(t) \leq -2\alpha_N y_N(t)^2 + 2\beta_N(t) y_N(t).
    \end{equation*}
    For $y_N(t) > 0$, we then have $\dot y_N(t) \leq -\alpha_N y_N(t) + \beta_N(t)$. 
    
    We now introduce the function $f_N(t, s) := -\alpha_N s + \beta_N(t)$. The solution of the differential equation $\dot s(t) = f_N\big(t, s(t) \big)$, $s(0) = \|x_N(0)\|_{P_N}$ is $s(t) = e^{-\alpha_N t} \left( \|x_N(0)\|_{P_N} + \int_{0}^t e^{\alpha_N \tau} \beta_N(\tau)\, d\tau \right)$. Since $f_N(t, s)$ is Lipschitz in $s$ and continuous in $t$, $\dot y_N(t) \leq f_N\big( t, y_N(t) \big)$ and $y_N(0) = z(0)$, the Comparison Lemma of \citep{Khalil} states that $y_N(t) \leq s(t)$ for all $t \geq 0$, hence \eqref{eq: x_N bound incomplete} holds.
\end{proof}

Note that the definition of $z_{max}^{P_N}$ in Proposition~\ref{prop: steady state xN} implies that $w_N$ is chosen first and the controller $u_N$ reacts optimally to it. The objective function not being concave-convex, there is an information imbalance giving an advantage to the second player. If we wanted instead the undesirable input to react optimally to any controller, we could use $z' := \underset{u_N\, \in\, \mathcal{U}_N}{\min} \, \underset{w_N\, \in\, \mathcal{W}_N}{\max} \big\|C_N w_N + B_N u_N\big\|_{P_N}$ in place of $z_{max}^{P_N}$. Note that $z' \geq z_{max}^{P_N}$.

Since $A_N$ is Hurwitz, we can bound the steady-state value of the state $x_N$ despite undesirable inputs that cannot be counteracted. We will now study the impact of $x_N$ on the rest of the network, whose dynamics follow \eqref{eq: X N-1}.
Recall that $\chi = \big( x_1, \hdots, x_{N-1} \big)$ and $\hat{D}$ was defined in \eqref{eq: notation D-}.
Dynamics~\eqref{eq: X N-1} are perturbed by the term $D_{-,N} x_N(t)$ bounded in Proposition~\ref{prop: steady state xN}. We can then evaluate how term $D_{-,N} x_N(t)$ impacts $\chi(t)$ by building on Proposition~\ref{prop: steady state xN} and reusing $P_N$, $Q_N$, $\alpha_N$, $\beta_N$ and $z_{max}^{P_N}$.
We will first investigate the scenario where $\hat{B}$ is full rank before requiring only controllability of pair $\big( \hat{A} + \hat{D}, \hat{B}\big)$.

\subsection{Fully-actuated networks}\label{subsec: fully-actuated}

In this section we assume that the combined control matrix of the first $N-1$ subsystems $\hat{B}$ is full rank.

\begin{proposition}\label{prop: steady state X}
    If $\hat{A} + \hat{D}$ and $A_N$ are Hurwitz, $\hat{B}$ is full rank, and $C_N \mathcal{W}_N \nsubseteq B_N \mathcal{U}_N$, then for any $\hat{P} \succ 0$ and $\hat{Q} \succ 0$ such that $(\hat{A}+\hat{D})^\top \hat{P} + \hat{P} (\hat{A}+\hat{D}) = -\hat{Q}$ let us define the constants $b_{min}^{\hat{P}} := \underset{\hat{u}\, \in\, \partial \hat{\mathcal{U}}}{\min} \big\{ \|\hat{B} \hat{u}\|_{\hat{P}} \big\}$, $\alpha := \frac{\lambda_{min}^{\hat{Q}}}{2\lambda_{max}^{\hat{P}}}$, $\gamma := \sqrt{ \frac{\max\lambda(D_{-,N}^\top \hat{P} D_{-,N})}{\lambda_{min}^{P_N}} }$ and $\gamma_N := \sqrt{ \frac{ \max \lambda(D_{N,-}^\top P_N D_{N,-})}{\lambda_{min}^{\hat{P}}} }$. \break
    If $\alpha \alpha_N \neq \gamma \gamma_N$, then there exist $h_\pm \in \mathbb{R}$ and $r_\pm \in \mathbb{R}$ such that
     \begin{equation}\label{eq: X bound}
        \|\chi(t)\|_{\hat{P}} \leq \max\left\{0,\ \frac{\gamma z_{max}^{P_N} - \alpha_N b_{min}^{\hat{P}}}{ \alpha \alpha_N - \gamma \gamma_N} + h_+ e^{(r_+ - \alpha_N)t} + h_- e^{(r_- - \alpha_N)t} \right\}.
    \end{equation}
    If $\alpha \alpha_N = \gamma \gamma_N$, there are constants $h_\pm \in \mathbb{R}$ such that
    \begin{equation}\label{eq: X bound degenerate}
        \|\chi(t)\|_{\hat{P}} \leq \max\left\{0,\ \frac{\gamma z_{max}^{P_N} - \alpha_N b_{min}^{\hat{P}}}{\alpha + \alpha_N} t + h_+ + h_- e^{-(\alpha + \alpha_N) t} \right\}.
    \end{equation}
\end{proposition}
\begin{proof}
    Since $\hat{A} + \hat{D}$ is Hurwitz, there exist a symmetric $\hat{P} \succ 0$ and $\hat{Q} \succ 0$ such that $(\hat{A} + \hat{D})^\top \hat{P} + \hat{P} (\hat{A} + \hat{D}) = -\hat{Q}$ according to Lyapunov theory \citep{Kalman}. Following the same steps as in the proof of Proposition~\ref{prop: steady state xN} with $\chi^\top \hat{P} \chi = \|\chi\|_{\hat{P}}^2$ and $\chi$ following the dynamics \eqref{eq: X N-1}, we first obtain
    \begin{equation*}
        \frac{d}{dt} \|\chi(t) \|_{\hat{P}}^2 = \chi^\top \big((\hat{A}+\hat{D})^\top \hat{P} + \hat{P} (\hat{A}+\hat{D}) \big) \chi + 2 \chi^\top \hat{P} \hat{B} \hat{u} + 2 \chi^\top \hat{P} D_{-,N} x_N.
    \end{equation*}
    Because $\hat{B}$ is full rank, for all $\chi(t) \neq 0$ there exist $\hat{u}(t) \in \hat{\mathcal{U}}$ such that $\hat{B} \hat{u}(t) = - \frac{\chi(t)}{\|\chi(t)\|_{\hat{P}}} b_{min}^{\hat{P}}$, as shown in the proof of Proposition~3 of \citep{ECC}. Then,
    \begin{equation*}
        \chi(t)^\top \hat{P} \hat{B} \hat{u}(t) =  \frac{-\chi(t)^\top \hat{P} \chi(t)}{\|\chi(t)\|_{\hat{P}}}b_{min}^{\hat{P}} = -\|\chi(t)\|_{\hat{P}} b_{min}^{\hat{P}}.
    \end{equation*}
    Since $\|\cdot\|_{\hat{P}}$ is a norm, it verifies the Cauchy-Schwarz inequality \citep{matrix_computations} $\chi^\top \hat{P} D_{-,N} x_N \leq \|\chi\|_{\hat{P}} \|D_{-,N} x_N\|_{\hat{P}}$. Then,
    \begin{equation*}
        \frac{d}{dt} \|\chi(t) \|_{\hat{P}}^2 \leq -\chi(t)^\top \hat{Q} \chi(t) - 2\|\chi(t)\|_{\hat{P}} b_{min}^{\hat{P}} + 2 \|\chi(t)\|_{\hat{P}} \|D_{-,N} x_N(t)\|_{\hat{P}}.
    \end{equation*}
    Because $\hat{P} \succ 0$ and $\hat{Q} \succ 0$, we obtain $-\chi^\top \hat{Q} \chi \leq -\frac{\lambda_{min}^{\hat{Q}}}{\lambda_{max}^{\hat{P}}} \|\chi\|_{\hat{P}}^2$. 
    Since $\hat{P}^\top = \hat{P} \succ 0$ and $\hat{P}_N \succ 0$, Lemma~\ref{lemma: P norm change} states $\|D_{-,N} x_N\|_{\hat{P}} \leq \gamma \|x_N\|_{P_N}$.
    We now combine these inequalities into
    \begin{align*}
        \frac{d}{dt} \|\chi(t) \|_{\hat{P}}^2 \leq -\frac{\lambda_{min}^{\hat{Q}}}{\lambda_{max}^{\hat{P}}} \|\chi(t)\|_{\hat{P}}^2 + 2 \|\chi(t)\|_{\hat{P}} \left( \gamma \|x_N(t)\|_{P_N} - b_{min}^{\hat{P}} \right).
    \end{align*}
    Following Proposition~\ref{prop: steady state xN}, we also include bound \eqref{eq: x_N bound incomplete} on $\|x_N(t)\|_{P_N}$, which yields
    \begin{equation*}
        \frac{d}{dt} \|\chi(t) \|_{\hat{P}}^2 \leq -\frac{\lambda_{min}^{\hat{Q}}}{\lambda_{max}^{\hat{P}}} \|\chi(t)\|_{\hat{P}}^2 + 2 \|\chi(t)\|_{\hat{P}} \left( \gamma e^{-\alpha_N t} \left( \|x_N(0)\|_{P_N} + \int_0^t \hspace{-1mm} e^{\alpha_N \tau} \beta_N(\tau) d\tau  \right) - b_{min}^{\hat{P}}  \right).
    \end{equation*}
    Define $y(t) := \|\chi(t)\|_{\hat{P}}$ and $y_N(t) := \|x_N(t)\|_{P_N}$. Since $P_N^\top = P_N \succ 0$ and $\hat{P} \succ 0$, Lemma~\ref{lemma: P norm change} states $\|D_{N,-} \chi\|_{P_N} \leq \gamma_N \|\chi\|_{\hat{P}}$, which can be used in $\beta_N$ defined in Proposition~\ref{prop: steady state xN} as
    \begin{equation}\label{eq: beta_N}
        \beta_N(\tau) = z_{max}^{P_N} + \|D_{N,-} \chi(\tau)\|_{P_N} \leq z_{max}^{P_N} + \gamma_N \|\chi(\tau)\|_{\hat{P}} = z_{max}^{P_N} + \gamma_N y(\tau).
    \end{equation}
    We notice that $\frac{d}{dt} \|\chi(t) \|_{\hat{P}}^2 = 2 y(t) \dot y(t)$, which yields
    \begin{equation*}
        2 y(t) \dot y(t) \leq -2\alpha y(t)^2 + 2y(t)\left( \gamma y_N(0) e^{-\alpha_N t} + \gamma  e^{-\alpha_N t} \int_0^t \hspace{-1mm} e^{\alpha_N \tau} \big( z_{max}^{P_N} + \gamma_N y(\tau) \big)\, d\tau - b_{min}^{\hat{P}}\right).
    \end{equation*}
    For $y(t) > 0$ we can divide both sides of the inequality by $2y(t)$ and we calculate the following trivial integral
    \begin{equation*}
        e^{-\alpha_N t} \int_0^t \hspace{-1mm} e^{\alpha_N \tau}\, d\tau = e^{-\alpha_N t} \frac{e^{\alpha_N t} -1}{\alpha_N} = \frac{1 - e^{-\alpha_N t}}{\alpha_N},
    \end{equation*}
    so that the differential inequality becomes
    \begin{align*}
        \dot y(t) &\leq -\alpha y(t) + \gamma y_N(0) e^{-\alpha_N t} + \gamma z_{max}^{P_N} \frac{1 - e^{-\alpha_N t}}{\alpha_N} + \gamma \gamma_N e^{-\alpha_N t}  \int_0^t \hspace{-1mm} e^{\alpha_N \tau}   y(\tau)\, d\tau - b_{min}^{\hat{P}}, \\
        &\leq -\alpha y(t) + \frac{\gamma z_{max}^{P_N}}{\alpha_N} - b_{min}^{\hat{P}} + \gamma \left( y_N(0) - \frac{z_{max}^{P_N}}{\alpha_N} \right) e^{-\alpha_N t} + \gamma \gamma_N e^{-\alpha_N t} \int_0^t \hspace{-1mm} e^{\alpha_N \tau} y(\tau)\, d\tau.
    \end{align*}
    Now multiply both sides by $e^{\alpha_N t} > 0$ and define $v(t) = e^{\alpha_N t} y(t)$. Then, $\dot v(t) = \alpha_N v(t) + e^{\alpha_N t} \dot y(t)$, which leads to
    \begin{equation*}
        e^{\alpha_N t} \dot y(t) = \dot v(t) - \alpha_N v(t) \leq -\alpha v(t) + \left( \frac{\gamma z_{max}^{P_N}}{\alpha_N} - b_{min}^{\hat{P}} \right) e^{\alpha_N t} + \gamma \left( y_N(0) - \frac{z_{max}^{P_N}}{\alpha_N} \right) + \gamma \gamma_N \int_0^t \hspace{-1mm} v(\tau)\, d\tau.
    \end{equation*}
    We introduce the function 
    \begin{equation*}
        f\big(t, s(t) \big) := (\alpha_N - \alpha) s(t) + \left( \frac{\gamma z_{max}^{P_N}}{\alpha_N} - b_{min}^{\hat{P}} \right) e^{\alpha_N t} + \gamma \left( y_N(0) - \frac{z_{max}^{P_N}}{\alpha_N} \right) + \gamma \gamma_N \int_0^t \hspace{-1mm} s(\tau)\, d\tau,
    \end{equation*}
    so that $\dot v(t) \leq f\big( t, v(t) \big)$. Now we search for a solution to the differential equation $\dot s(t) = f\big(t, s(t) \big)$. Differentiating this equation yields
    \begin{align*}
        \ddot s(t) = \frac{d}{dt} f\big(t, s(t) \big) = (\alpha_N - \alpha) \dot s(t) + \left( \frac{\gamma z_{max}^{P_N}}{\alpha_N} - b_{min}^{\hat{P}} \right) \alpha_N e^{\alpha_N t} + 0 + \gamma \gamma_N s(t),
    \end{align*}
    i.e.,
    \begin{equation}\label{eq: non-homogeneous ODE}
        \ddot s(t) + (\alpha - \alpha_N) \dot s(t) - \gamma \gamma_N s(t) - \left( \gamma z_{max}^{P_N} - \alpha_N b_{min}^{\hat{P}} \right) e^{\alpha_N t} = 0.
    \end{equation}
    We need to distinguish two cases when solving this differential equation as detailed in Lemma~\ref{lemma: ODE calculations}. If $\alpha \alpha_N \neq \gamma \gamma_N$, then the solution of \eqref{eq: non-homogeneous ODE} is
    \begin{equation*}
        s(t) = p e^{\alpha_N t} + h_+ e^{r_+ t} + h_- e^{r_- t}
    \end{equation*}
    with
    \begin{equation}\label{eq: r_pm}
        p = \frac{\gamma z_{max}^{P_N} - \alpha_N b_{min}^{\hat{P}}}{ \alpha \alpha_N - \gamma \gamma_N}, \qquad r_{\pm} = \frac{1}{2} \Big( \alpha_N - \alpha \pm \sqrt{(\alpha - \alpha_N)^2 + 4\gamma \gamma_N} \Big),
    \end{equation}
    and $h_\pm \in \mathbb{R}$ are two constants to be determined. Now we can apply the Comparison Lemma of \citep{Khalil} stating that if $\dot s(t) = f\big( t, s(t) \big)$, $f$ is continuous in $t$ and locally Lipschitz in $s$ and $s(0) = v(0)$, then $\dot v(t) \leq f\big( t, v(t) \big)$ implies $v(t) \leq s(t)$ for all $t \geq 0$. Using $\|\chi(t)\|_{\hat{P}} = y(t) = e^{-\alpha_N t} v(t) \leq e^{-\alpha_N t} s(t)$, we finally obtain \eqref{eq: X bound}.
    To determine $h_\pm$, we use the initial conditions $s(0) = v(0) = y(0)$ and $\dot s(0) = f \big( 0, s(0) \big)$, which yield 
    \begin{equation*}
        h_{\pm} = \frac{(\alpha_N - \alpha - r_{\mp})\|\chi(0)\|_{\hat{P}} + \gamma \|x_N(0)\|_{P_N} -b_{min}^{\hat{P}} + (r_{\mp} - \alpha_N)p}{\pm\sqrt{(\alpha - \alpha_N)^2 + 4\gamma \gamma_N}}.
    \end{equation*}
    
    In the case $\alpha \alpha_N = \gamma \gamma_N$, the solution of \eqref{eq: non-homogeneous ODE} is
    \begin{equation*}
        s(t) = p t e^{\alpha_N t} + h_+ e^{\alpha_N t} + h_- e^{-\alpha t}
    \end{equation*}
    with
    \begin{align*}
        p = \frac{\gamma z_{max}^{P_N} - \alpha_N b_{min}^{\hat{P}}}{\alpha + \alpha_N}, \quad \text{and} \quad h_{\pm} = \frac{\frac{1}{2}\big(-\alpha_N - \alpha \pm 3(\alpha - \alpha_N) \big)\|\chi(0)\|_{\hat{P}} \mp \gamma \|x_N(0)\|_{P_N} \pm b_{min}^{\hat{P}} \pm p}{\alpha_N + \alpha},
    \end{align*}
    obtained from the initial conditions $s(0) = y(0)$ and $\dot s(0) = f\big(0, s(0) \big)$. Applying the Comparison Lemma of \citep{Khalil} as above, we obtain $\|\chi(t)\|_{\hat{P}} = y(t) = e^{-\alpha_N t} v(t) \leq e^{-\alpha_N t} s(t)$, which yields \eqref{eq: X bound degenerate}.
\end{proof}

We can now derive conditions for subsystem~\eqref{eq: X N-1} to be resiliently stabilizable despite the perturbations created by $x_N$. These conditions solve Problem~\ref{prob: nonresilient stabilizability} in the fully-actuated network scenario.

\begin{theorem}\label{thm: X stabilizable}
    If $\hat{A} + \hat{D}$ and $A_N$ are Hurwitz, $\hat{B}$ is full rank, and $C_N \mathcal{W}_N \nsubseteq B_N \mathcal{U}_N$, $\gamma \gamma_N \leq \alpha \alpha_N$ and $\gamma z_{max}^{P_N} < \alpha_N b_{min}^{\hat{P}}$, then subsystem~\eqref{eq: X N-1} is resiliently stabilizable in finite time.
\end{theorem}
\begin{proof}
    Let us first consider the case $\gamma \gamma_N = \alpha \alpha_N$. Since $\alpha > 0$ and $\alpha_N > 0$, the exponential term in \eqref{eq: X bound degenerate} goes to zero asymptotically.
    By assumption $\gamma z_{max}^{P_N} - \alpha_N b_{min}^{\hat{P}} < 0$ and $\alpha + \alpha_N > 0$, so the ratio of these factors is negative. Because this ratio is multiplied by $t$ in \eqref{eq: X bound degenerate}, there exists some time $T \geq 0$ such that for all $t \geq T$
    \begin{equation}\label{eq: linear decrease}
        \frac{\gamma z_{max}^{P_N} - \alpha_N b_{min}^{\hat{P}}}{\alpha + \alpha_N} t + h_+ + h_- e^{-(\alpha+\alpha_N)t} \leq 0.
    \end{equation}
    Therefore, according to \eqref{eq: X bound degenerate}, subsystem \eqref{eq: X N-1} is resiliently stabilizable in finite time.
    
    Now consider the case $\gamma \gamma_N < \alpha \alpha_N$. We will show that this inequality is equivalent to $r_+ - \alpha_N < 0$, where $r_+$ is defined in \eqref{eq: r_pm}. Indeed,
    \begin{align*}
        r_+ - \alpha_N < 0  \iff & -\frac{1}{2}(\alpha_N + \alpha) + \frac{1}{2}\sqrt{(\alpha - \alpha_N)^2 + 4 \gamma \gamma_N} < 0 \iff (\alpha - \alpha_N)^2 + 4 \gamma \gamma_N < (\alpha_N + \alpha)^2 \\
        \iff & -2 \alpha \alpha_N + 4\gamma \gamma_N < 2\alpha \alpha_N \iff \gamma \gamma_N < \alpha \alpha_N.
    \end{align*}
    Since $r_- \leq r_+$, we also have $r_- - \alpha_N < 0$, so both exponential terms in \eqref{eq: X bound} converge to zero. Additionally, the fraction term in \eqref{eq: X bound} is negative, so the right-hand side of \eqref{eq: X bound} reaches zero in finite time. Therefore, subsystem~\eqref{eq: X N-1} is resiliently stabilizable in finite time.
\end{proof}

Let us now give some intuition concerning Theorem~\ref{thm: X stabilizable}.
Since $\gamma$ is proportional to the norm of the matrix $D_{-,N}$ which multiplies $x_N(t)$ in \eqref{eq: X N-1}, $\gamma$ quantifies the impact of nonresilient subsystem~\eqref{eq:split system N} of state $x_N(t)$ on the rest of the network~\eqref{eq: X N-1} of state $\chi(t)$.
Reciprocally, $\gamma_N$ quantifies the impact of $\chi(t)$ on $x_N(t)$.
On the other hand, $\alpha = \frac{\lambda_{min}^{\hat{Q}}}{2\lambda_{max}^{\hat{P}}}$ relates to the joint stability of the first $N-1$ subsystems of network~\eqref{eq: X N-1}, while $\alpha_N$ relates to the stability of malfunctioning subsystem~\eqref{eq:split system N}.
Therefore, condition $\gamma\gamma_N \leq \alpha \alpha_N$ follows the intuition that the magnitude of the perturbations arising from the coupling between subsystems~\eqref{eq: X N-1} and \eqref{eq:split system N} must be weaker than the stability of each of these subsystems.

We will now discuss the other stabilizability condition of Theorem~\ref{thm: X stabilizable}, namely $\gamma z_{max}^{P_N} < \alpha_N b_{min}^{\hat{P}}$. Since $z_{max}^{P_N}$ describes the magnitude of the destabilizing inputs in subsystem~\eqref{eq:split system N}, term $\gamma z_{max}^{P_N}$ quantifies the destabilizing influence of $w_N$ on the state of the rest of the network $\chi$. On the other hand, $b_{min}^{\hat{P}}$ relates to the magnitude of the stabilizing inputs in subsystem~\eqref{eq: X N-1} and $\alpha_N$ relates to the Hurwitzness of malfunctioning subsystem~\eqref{eq:split system N}.
Therefore, condition $\gamma z_{max}^{P_N} < \alpha_N b_{min}^{\hat{P}}$ carries the intuition that the stabilizing terms of the network must overcome the destabilizing ones.

Theorem~\ref{thm: X stabilizable} can also be used in an adversarial fashion, by identifying subsystems of the network which are not guaranteed to be resiliently stabilizable by Theorem~\ref{thm: X stabilizable}.

Since we have bounded the state $\chi$ of the first $N-1$ subsystems, we can now derive a closed-form bound on the state $x_N$ of the malfunctioning subsystem $N$. Indeed, the bound on $x_N$ derived in Proposition~\ref{prop: steady state xN} depends on $\chi(t)$ through the term $\beta_N(t)$.

\begin{proposition}\label{prop: full steady state xN}
    If $\hat{A} + \hat{D}$ and $A_N$ are Hurwitz, $\hat{B}$ is full rank, and $C_N \mathcal{W}_N \nsubseteq B_N \mathcal{U}_N$, we can bound the state of subsystem~\eqref{eq:split system N} as
    \begin{equation}\label{eq: x_N bound}
        \|x_N(t)\|_{P_N} \leq \left\{\def\arraystretch{2}\begin{array}{l}
            \max\left\{ 0,\ e^{-\alpha_N t} \|x_N(0)\|_{P_N} + M(t) \right\} \qquad \text{if} \quad \|\chi(t)\|_{\hat{P}} > 0, \\
            \frac{z_{max}^{P_N}}{\alpha_N} + \left( \|x_N(0)\|_{P_N} - \frac{z_{max}^{P_N}}{\alpha_N} \right) e^{-\alpha_N t} \qquad \text{otherwise,}
        \end{array} \right.
    \end{equation}
    with 
    \begin{equation}\label{eq: M}
        M(t) = \hspace{-1mm} \left\{ \hspace{-2mm} \def\arraystretch{2}\begin{array}{l}
             \frac{\alpha z_{max}^{P_N} - \gamma_N b_{min}^{\hat{P}} }{\alpha \alpha_N - \gamma \gamma_N}\big(1 - e^{-\alpha_N t} \big) + e^{-\alpha_N t} \left( \frac{\gamma_N h_+}{r_+}\big( e^{r_+ t} -1 \big) + \frac{\gamma_N h_-}{r_-} \big( e^{r_- t} -1\big)  \right) \quad \text{if} \ \alpha \alpha_N \neq \gamma \gamma_N,\\
             \frac{1-e^{-\alpha_N t}}{\alpha_N}\Big(\gamma_N h_+ + \frac{\alpha_N z_{max}^{P_N} + \gamma_N b_{min}^{\hat{P}}}{\alpha+\alpha_N} \Big) + \frac{\alpha z_{max}^{P_N} - \gamma_N b_{min}^{\hat{P}}}{\alpha + \alpha_N} t + \frac{\gamma_N h_-}{\alpha}\big(1 - e^{-\alpha t}\big)e^{-\alpha_N t} \hspace{2mm} \text{otherwise}.
        \end{array} \right. 
    \end{equation}
\end{proposition}
\begin{proof}
    We recall from Proposition~\ref{prop: steady state xN} that $\|x_N(t)\|_{P_N} \leq e^{-\alpha_N t} \left( \|x_N(0)\|_{P_N} + \int_0^t e^{\alpha_N \tau} \beta_N(\tau)\, d\tau \right)$ \eqref{eq: x_N bound incomplete}.
    Following \eqref{eq: beta_N}, we have $\beta_N(t) \leq z_{max}^{P_N} + \gamma_N \|\chi(t)\|_{\hat{P}}$. We can bound $\|\chi(t)\|_{\hat{P}}$ with \eqref{eq: X bound} or \eqref{eq: X bound degenerate} from Proposition~\ref{prop: steady state X} depending on the values of $\alpha \alpha_N$ and $\gamma \gamma_N$.
    
    We start with the case where $\alpha \alpha_N \neq \gamma \gamma_N$ and $\|\chi(t)\|_{\hat{P}} > 0$. Then, bound \eqref{eq: X bound} combined with \eqref{eq: beta_N} yields
    \begin{align*}
        \int_0^t e^{\alpha_N \tau} \beta_N(\tau)\, d\tau &\leq \int_0^t e^{\alpha_N \tau} \Big( z_{max}^{P_N} + \gamma_N p + \gamma_N h_+ e^{(r_+ - \alpha_N)\tau} + \gamma_N h_- e^{(r_- - \alpha_N)\tau} \Big) d\tau \\
        &= \frac{e^{\alpha_N t} -1}{\alpha_N}\big( z_{max}^{P_N} + \gamma_N p\big) + \frac{\gamma_N h_+}{r_+}\big( e^{r_+ t} -1 \big) + \frac{\gamma_N h_-}{r_-} \big( e^{r_- t} -1\big).
    \end{align*}
    We replace $p$ in
    \begin{equation*}
        \frac{z_{max}^{P_N} + \gamma_N p}{\alpha_N} = \frac{\alpha z_{max}^{P_N} - \gamma_N b_{min}^{\hat{P}} }{\alpha \alpha_N - \gamma \gamma_N} \qquad \text{with} \qquad p = \frac{\gamma z_{max}^{P_N} - \alpha_N b_{min}^{\hat{P}}}{ \alpha \alpha_N - \gamma \gamma_N}.
    \end{equation*}
    Then, plugging the integral calculated above in \eqref{eq: x_N bound incomplete}, we obtain
    \begin{equation*}
        \|x_N(t)\|_{P_N} \hspace{-1mm} \leq e^{-\alpha_N t} \hspace{-1mm} \left( \hspace{-1mm} \|x_N(0)\|_{P_N} \hspace{-1mm} + \frac{\alpha z_{max}^{P_N} - \gamma_N b_{min}^{\hat{P}} }{\alpha \alpha_N - \gamma \gamma_N}\big(e^{\alpha_N t} -1\big) + \frac{\gamma_N h_+}{r_+}\big( e^{r_+ t} -1 \big) + \frac{\gamma_N h_-}{r_-} \big( e^{r_- t} -1\big)  \hspace{-1mm} \right) \hspace{-1mm},
    \end{equation*}
   which yields \eqref{eq: x_N bound}.
    
    When $\|\chi(t)\|_{\hat{P}} = 0$, bound \eqref{eq: X bound} is modified and yields
    \begin{align*}
        \|x_N(t)\|_{P_N} &\leq e^{-\alpha_N t} \left( \|x_N(0)\|_{P_N} + \int_0^t e^{\alpha_N \tau} z_{max}^{P_N}\, d\tau \right) = \frac{z_{max}^{P_N}}{\alpha_N} + \left( \|x_N(0)\|_{P_N} - \frac{z_{max}^{P_N}}{\alpha_N} \right) e^{-\alpha_N t}.
    \end{align*}

    We can now address the other case where $\alpha \alpha_N = \gamma \gamma_N$ and $\|\chi(t)\|_{\hat{P}} > 0$ is bounded by \eqref{eq: X bound degenerate}, which yields
    \begin{align*}
        \int_0^t e^{\alpha_N \tau} \beta_N(\tau)\, d\tau &\leq \int_0^t e^{\alpha_N \tau} \big( z_{max}^{P_N} + \gamma_N p \tau + \gamma_N h_+ + \gamma_N h_- e^{-(\alpha + \alpha_N)\tau} \big) d\tau \\
        &= \big(z_{max}^{P_N} + \gamma_N h_+\big) \int_0^t e^{\alpha_N \tau} \, d\tau + \gamma_N p \int_0^t \tau e^{\alpha_N \tau} \, d\tau + \gamma_N h_- \int_0^t e^{-\alpha \tau} \, d\tau \\
        &= \big(z_{max}^{P_N} + \gamma_N h_+\big) \frac{e^{\alpha_N t} -1}{\alpha_N} + \frac{\gamma_N p}{\alpha_N} \left( \frac{1 - e^{\alpha_N t}}{\alpha_N} + t e^{\alpha_N t} \right) + \gamma_N h_- \frac{e^{-\alpha t} -1}{-\alpha} \\
        &= \frac{e^{\alpha_N t} -1}{\alpha_N}\Big(z_{max}^{P_N} + \gamma_N h_+ - \frac{\gamma_N p}{\alpha_N} \Big) + \frac{\gamma_N p}{\alpha_N}t e^{\alpha_N t} + \frac{\gamma_N h_-}{\alpha}\big(1 - e^{-\alpha t} \big),
    \end{align*}
    where we calculated $\int_0^t \tau e^{\alpha_N \tau} \, d\tau$ using
    \begin{align*}
        \int_0^t \frac{d}{d\tau} \frac{\tau e^{\alpha_N \tau}}{\alpha_N} d\tau &= \Bigg[ \frac{\tau e^{\alpha_N \tau}}{\alpha_N} \Bigg]_0^t =  \frac{t e^{\alpha_N t}}{\alpha_N} - 0 = \int_0^t \frac{e^{\alpha_N \tau}}{\alpha_N} d\tau + \int_0^t \tau e^{\alpha_N \tau} d\tau \\
        &= \frac{e^{\alpha_N t}-1}{\alpha_N^2} + \int_0^t \tau e^{\alpha_N \tau} d\tau.
    \end{align*}
    
    Then, we replace $p$ with its definition:
    \begin{equation*}
        \frac{\gamma_N p}{\alpha_N} = \frac{\gamma_N \gamma z_{max}^{P_N} - \gamma_N \alpha_N b_{min}^{\hat{P}}}{\alpha_N(\alpha + \alpha_N)} = \frac{\alpha z_{max}^{P_N} - \gamma_N b_{min}^{\hat{P}}}{\alpha + \alpha_N} \qquad \text{thanks to} \qquad p = \frac{\gamma z_{max}^{P_N} - \alpha_N b_{min}^{\hat{P}}}{\alpha + \alpha_N}
    \end{equation*}
    and $\gamma \gamma_N = \alpha \alpha_N$. 
    Multiplying the upper bound on $\int_0^t e^{\alpha_N \tau} \beta_N(\tau)d\tau$ calculated previously by $e^{-\alpha_N t}$ yields
    \begin{align*}
        \int_0^t \hspace{-2mm} e^{\alpha_N (\tau-t)} \beta_N(\tau) d\tau &\leq \frac{1-e^{-\alpha_N t}}{\alpha_N}\Bigg( z_{max}^{P_N} + \gamma_N h_+ - \frac{\alpha z_{max}^{P_N} - \gamma_N b_{min}^{\hat{P}}}{\alpha + \alpha_N} \Bigg) + \frac{\alpha z_{max}^{P_N} - \gamma_N b_{min}^{\hat{P}}}{\alpha + \alpha_N} t \\
        &\quad + \frac{\gamma_N h_-}{\alpha}\big(1 - e^{-\alpha t}\big)e^{-\alpha_N t} \\
        & \hspace{-10mm} \leq \frac{1-e^{-\alpha_N t}}{\alpha_N}\Bigg( \gamma_N h_+ + \frac{\alpha_N z_{max}^{P_N} + \gamma_N b_{min}^{\hat{P}}}{\alpha + \alpha_N} \Bigg) + \frac{\alpha z_{max}^{P_N} - \gamma_N b_{min}^{\hat{P}}}{\alpha + \alpha_N} t + \frac{\gamma_N h_-}{\alpha}\big(1 - e^{-\alpha t}\big)e^{-\alpha_N t}
    \end{align*}
    We finally obtain \eqref{eq: x_N bound} thanks to \eqref{eq: x_N bound incomplete}.
\end{proof}

\begin{remark}\label{rmk: discontinuous switch}
    The switch between bounds \eqref{eq: x_N bound} is likely to be discontinuous. Indeed, the bound in \eqref{eq: x_N bound} valid for $\|\chi(t)\|_{\hat{P}} > 0$ relies on all the overapproximations of bound \eqref{eq: X bound}, whereas the case $\|\chi(t)\|_{\hat{P}} = 0$ is derived without these overapproximations.
\end{remark}

Thanks to Propositions~\ref{prop: steady state X} and \ref{prop: full steady state xN}, we now have a complete description of the network state after a nonresilient loss of control authority. These two results relied on the full rank assumption of $\hat{B}$, the control matrix of the unaffected part of the network. Because this assumption might be too restrictive, we will now employ a different approach to bound the states of an underactuated network.

\subsection{Underactuated networks}\label{subsec: nonresilient underactuated}

Let us now assume that $\hat{B}$ is not full rank, which prevents the use of Proposition~\ref{prop: steady state X}. Instead of the stabilizing control input of constant magnitude $\hat{B}\hat{u}(t) = - \frac{\chi(t)}{\|\chi(t)\|_{\hat{P}}} b_{min}^{\hat{P}}$ used in Proposition~\ref{prop: steady state X}, we will employ a linear control to bound network state $\chi$.

If pair $\big(\hat{A} + \hat{D}, \hat{B}\big)$ is controllable, there exist a matrix $K$ such that $\hat{A} + \hat{D} - \hat{B}K$ is Hurwitz. Then, for any $\hat{P} \succ 0$ and $\hat{Q} \succ 0$ such that $(\hat{A}+\hat{D}- \hat{B}K)^\top \hat{P} + \hat{P} (\hat{A}+\hat{D}- \hat{B}K) = -\hat{Q}$, we can define the same constants as in Proposition~\ref{prop: steady state X}, namely $\alpha = \frac{\lambda_{min}^{\hat{Q}}}{2\lambda_{max}^{\hat{P}}}$, $\gamma = \sqrt{ \frac{\lambda_{max}^{D_{-,N}^\top \hat{P} D_{-,N}}}{\lambda_{min}^{P_N}} }$ and $\gamma_N = \sqrt{ \frac{\lambda_{max}^{D_{N,-}^\top P_N D_{N,-}}}{\lambda_{min}^{\hat{P}}} }$.

\begin{proposition}\label{prop: steady state X underactuated}
    If pair $\big(\hat{A} + \hat{D}, \hat{B}\big)$ is controllable, $A_N$ is Hurwitz, $C_N \mathcal{W}_N \nsubseteq B_N \mathcal{U}_N$, $\gamma \gamma_N < \alpha \alpha_N$ and $\underset{t\, \geq\, 0}{\sup}\ b(t) \leq \frac{\sqrt{\lambda_{min}^{\hat{P}}}}{\|K\|}$, then system~\eqref{eq: X N-1} is resiliently bounded: $\|\chi(t)\|_{\hat{P}} \leq \max\left\{0,\ b(t) \right\}$ for all $t \geq 0$, with
    \begin{equation}\label{eq: b(t)}
         b(t) := p + h_+ e^{(r_+ - \alpha_N)t} + h_- e^{(r_- - \alpha_N)t}, 
    \end{equation}
    \begin{equation*}
        h_{\pm} = \frac{(\alpha_N - \alpha - r_{\mp})\|\chi(0)\|_{\hat{P}} + \gamma \|x_N(0)\|_{P_N} + (r_{\mp} - \alpha_N)p}{\pm\sqrt{(\alpha_N - \alpha)^2 + 4\gamma \gamma_N}} \quad \text{and} \quad p = \frac{\gamma z_{max}^{P_N}}{ \alpha \alpha_N - \gamma \gamma_N}.
    \end{equation*}
\end{proposition}
\begin{proof}
    Since pair $\big(\hat{A} + \hat{D}, \hat{B}\big)$ is controllable, there exists a matrix $K$ such that $\hat{A} + \hat{D} - \hat{B}K$ is Hurwitz \citep{Khalil}. Then, there exists $\hat{P} \succ 0$ and $\hat{Q} \succ 0$ such that $(\hat{A} + \hat{D}- \hat{B}K)^\top \hat{P} + \hat{P} (\hat{A} + \hat{D}- \hat{B}K) = -\hat{Q}$ according to Lyapunov theory \citep{Kalman}. We will follow the same steps as in the proof of Proposition~\ref{prop: steady state X} with $\chi^\top \hat{P} \chi = \|\chi\|_{\hat{P}}^2$ and $\chi$ following the dynamics \eqref{eq: X N-1} with $\hat{u}(t) = -K\chi(t)$. Once we obtain bounds on $\chi(t)$ we will verify under which conditions is $\hat{u}$ admissible. We first obtain
    \begin{equation*}
        \frac{d}{dt} \|\chi(t)\|_{\hat{P}}^2 = \chi(t)^\top \big((\hat{A}+\hat{D}- \hat{B}K)^\top \hat{P} + \hat{P} (\hat{A}+\hat{D}- \hat{B}K) \big) \chi(t) + 2 \chi(t)^\top \hat{P} D_{-,N} x_N(t).
    \end{equation*}
    We then proceed as in Proposition~\ref{prop: steady state X}, but without the term $b_{min}^{\hat{P}}$. Since $\|\cdot\|_{\hat{P}}$ is a norm, it verifies the Cauchy-Schwarz inequality \citep{matrix_computations} $\chi^\top \hat{P} D_{-,N} x_N \leq \|\chi(t)\|_{\hat{P}} \|D_{-,N} x_N(t)\|_{\hat{P}}$. Then,
    \begin{equation*}
        \frac{d}{dt} \|\chi(t)\|_{\hat{P}}^2 \leq -\chi(t)^\top \hat{Q} \chi(t) + 2 \|\chi(t)\|_{\hat{P}} \|D_{-,N} x_N(t)\|_{\hat{P}}.
    \end{equation*}
    Because $\hat{P} \succ 0$ and $\hat{Q} \succ 0$, we obtain $-\chi^\top \hat{Q} \chi \leq -\frac{\lambda_{min}^{\hat{Q}}}{\lambda_{max}^{\hat{P}}} \|\chi\|_{\hat{P}}^2$. Since $\hat{P}^\top = \hat{P} \succ 0$ and $\hat{P}_N \succ 0$, Lemma~\ref{lemma: P norm change} states $\|D_{-,N} x_N\|_{\hat{P}} \leq \gamma \|x_N\|_{P_N}$.
    We now combine these inequalities into
    \begin{align*}
        \frac{d}{dt} \|\chi(t)\|_{\hat{P}}^2 \leq -\frac{\lambda_{min}^{\hat{Q}}}{\lambda_{max}^{\hat{P}}} \|\chi(t)\|_{\hat{P}}^2 + 2 \|\chi(t)\|_{\hat{P}} \gamma \|x_N(t)\|_{P_N}.
    \end{align*}
    Following Proposition~\ref{prop: steady state xN}, we also include bound \eqref{eq: x_N bound incomplete} on $\|x_N(t)\|_{P_N}$, which yields
    \begin{equation*}
        \frac{d}{dt} \|\chi(t) \|_{\hat{P}}^2 \leq -\frac{\lambda_{min}^{\hat{Q}}}{\lambda_{max}^{\hat{P}}} \|\chi(t)\|_{\hat{P}}^2 + 2 \|\chi(t)\|_{\hat{P}} \gamma e^{-\alpha_N t} \left( \|x_N(0)\|_{P_N} + \int_0^t \hspace{-1mm} e^{\alpha_N \tau} \big(z_{max}^{P_N} + \|D_{N,-} \chi(\tau)\|_{P_N}\big) d\tau  \right).
    \end{equation*}
    Since $P_N^\top = P_N \succ 0$ and $\hat{P} \succ 0$, Lemma~\ref{lemma: P norm change} states $\|D_{N,-} \chi\|_{P_N} \leq \gamma_N \|\chi\|_{\hat{P}}$. Define $y(t) := \|\chi(t)\|_{\hat{P}}$ and $y_N(t) := \|x_N(t)\|_{P_N}$. We notice that $\frac{d}{dt} \|\chi(t) \|_{\hat{P}}^2 = 2 y(t) \dot y(t)$, which yields
    \begin{equation*}
        2 y(t) \dot y(t) \leq -2\alpha y(t)^2 + 2y(t)\left( \gamma y_N(0) e^{-\alpha_N t} + \gamma  e^{-\alpha_N t} \int_0^t \hspace{-1mm} e^{\alpha_N \tau} \big( z_{max}^{P_N} + \gamma_N y(\tau) \big)\, d\tau \right).
    \end{equation*}
    For $y(t) > 0$ we can divide both sides of the inequality by $2y(t)$ so that the differential inequality becomes
    \begin{align*}
        \dot y(t) &\leq -\alpha y(t) + \gamma y_N(0) e^{-\alpha_N t} + \gamma z_{max}^{P_N} \frac{1 - e^{-\alpha_N t}}{\alpha_N} + \gamma \gamma_N e^{-\alpha_N t}  \int_0^t \hspace{-1mm} e^{\alpha_N \tau}   y(\tau)\, d\tau, \\
        &\leq -\alpha y(t) + \frac{\gamma z_{max}^{P_N}}{\alpha_N} + \gamma \left( y_N(0) - \frac{z_{max}^{P_N}}{\alpha_N} \right) e^{-\alpha_N t} + \gamma \gamma_N e^{-\alpha_N t} \int_0^t \hspace{-1mm} e^{\alpha_N \tau} y(\tau)\, d\tau.
    \end{align*}
    Now multiply both sides by $e^{\alpha_N t} > 0$ and define $v(t) = e^{\alpha_N t} y(t)$. Then, $\dot v(t) = \alpha_N v(t) + e^{\alpha_N t} \dot y(t)$, which leads to
    \begin{equation*}
        e^{\alpha_N t} \dot y(t) = \dot v(t) - \alpha_N v(t) \leq -\alpha v(t) + \frac{\gamma z_{max}^{P_N}}{\alpha_N} e^{\alpha_N t} + \gamma \left( y_N(0) - \frac{z_{max}^{P_N}}{\alpha_N} \right) + \gamma \gamma_N \int_0^t \hspace{-1mm} v(\tau)\, d\tau.
    \end{equation*}
    We introduce the function 
    \begin{equation*}
        g\big(t, s(t) \big) := (\alpha_N - \alpha) s(t) + \frac{\gamma z_{max}^{P_N}}{\alpha_N} e^{\alpha_N t} + \gamma \left( y_N(0) - \frac{z_{max}^{P_N}}{\alpha_N} \right) + \gamma \gamma_N \int_0^t \hspace{-1mm} s(\tau)\, d\tau,
    \end{equation*}
    so that $\dot v(t) \leq g\big( t, v(t) \big)$. Now we search for a solution to the differential equation $\dot s(t) = g\big(t, s(t) \big)$. Differentiating this equation yields
    \begin{equation*}
        \ddot s(t) = \frac{d}{dt} g\big(t, s(t) \big) = (\alpha_N - \alpha) \dot s(t) + \gamma z_{max}^{P_N} e^{\alpha_N t} + 0 + \gamma \gamma_N s(t).
    \end{equation*}
    If $\alpha \alpha_N \neq \gamma \gamma_N$, then the solution to this differential equation is $s(t) = pe^{\alpha_N t} + h_+ e^{r_+ t} + h_- e^{r_- t}$, with $r_{\pm} = \frac{1}{2} \big( \alpha_N - \alpha \pm \sqrt{(\alpha - \alpha_N)^2 + 4\gamma \gamma_N} \big)$, $p = \frac{\gamma z_{max}^{P_N}}{ \alpha \alpha_N - \gamma \gamma_N} > 0$ and two constants $h_\pm \in \mathbb{R}$. Using the Comparison Lemma of \citep{Khalil} we obtain
    \begin{equation}\label{eq: X bound not full rank}
        \|\chi(t)\|_{\hat{P}} \leq \max\left\{ 0,\ p + h_+ e^{(r_+ - \alpha_N)t} + h_- e^{(r_- - \alpha_N) t} \right\},
    \end{equation}
    for all $t \geq 0$ and the initial conditions are $s(0) = y(0)$ and $\dot s(0) = g(0, s(0))$, i.e.,
    \begin{equation*}
        p + h_+ + h_- = \|\chi(0)\|_{\hat{P}} \quad \text{and} \quad \alpha_N p + h_+ r_+ + h_- r_- = (\alpha_N - \alpha) \|\chi(0)\|_{\hat{P}} + \gamma \|x_N(0)\|_{P_N}.
    \end{equation*}
    These initial conditions can be solved to determine $h_\pm$.
    
    Bound \eqref{eq: X bound not full rank} is only valid when $\hat{u}(t) = -K\chi(t) \in \mathcal{U} = [-1, 1]^m$. For this control law to be admissible, we then need $\|\chi(t)\|_{\hat{P}} \leq \frac{\sqrt{\lambda_{min}^{\hat{P}}}}{\|K\|}$ at all times $t \geq 0$, since $\| \hat{u}(t)\| \leq \|K\| \|\chi(t)\| \leq \|K\| \frac{\|\chi(t)\|_{\hat{P}}}{\sqrt{\lambda_{min}^{\hat{P}}}}$.
    Since $\|\chi(t)\|_{\hat{P}} \leq \underset{t\, \geq\, 0}{\sup}\ b(t) \leq \frac{\sqrt{\lambda_{min}^{\hat{P}}}}{\|K\|}$, $\hat{u}$ is admissible.
\end{proof}

\begin{remark}
    The condition $\gamma \gamma_N < \alpha \alpha_N$ in Proposition~\ref{prop: steady state X underactuated} is necessary for the boundedness of $\chi(t)$, which in turn guarantees the admissibility of the control law $\hat{u}(t) = -K\chi(t)$.
    
    Indeed, if $\gamma \gamma_N > \alpha \alpha_N$, then $r_+ - \alpha_N > 0$, which leads to the divergence of the corresponding exponential term in \eqref{eq: X bound not full rank} and hence $\chi(t)$ might not be bounded.
    
    On the other hand, if $\alpha \alpha_N = \gamma \gamma_N$, the same process as in Proposition~\ref{prop: steady state X underactuated} leads to $\|\chi(t)\|_{\hat{P}} \leq \max\left\{0,\ \frac{\gamma z_{max}^{P_N}}{\alpha + \alpha_N} t + h_+ + h_- e^{-(\alpha + \alpha_N) t}\right\}$ for all $t \geq 0$, for some constants $h_\pm \in \mathbb{R}$. The term linear in $t$ grows unbounded since $\gamma z_{max}^{P_N} > 0$ and $\alpha + \alpha_N > 0$. In this case, $\chi(t)$ might not be bounded.
\end{remark}

Note that the perturbation from subsystem~\eqref{eq:split system N} in norm bounds \eqref{eq: X bound not full rank} is modeled by term $z_{max}^{P_N} > 0$ of constant magnitude. Hence, this perturbation cannot be overcome by the linear control $\hat{u}(t) = -K\chi(t)$ when $\chi$ is near $0$. That is why Proposition~\ref{prop: steady state X underactuated} only guarantees the boundedness of $\chi$ and not its resilient stabilizability.

\begin{remark}\label{rmk: X0 = 0}
    If the network is initially at rest when the loss of control authority occurs, i.e., if $\chi(0) = 0$ and $x_N(0) = 0$, then $h_+ < 0$ and $h_- > 0$, so that $m = p + h_- = -h_+ = \frac{-(r_- + \alpha_N)\gamma z_{max}^{P_N}}{(\alpha \alpha_N - \gamma \gamma_N)\sqrt{(\alpha_N - \alpha)^2 + 4\gamma \gamma_N}}$.
\end{remark}

Using bound \eqref{eq: X bound not full rank} in \eqref{eq: x_N bound incomplete}, we can now derive a closed-form bound on $x_N$ as we did in Proposition~\ref{prop: full steady state xN} when $\hat{B}$ was full rank.

\begin{proposition}\label{prop: full steady state xN underactuated}
    If pair $\big(\hat{A} + \hat{D}, \hat{B}\big)$ is controllable, $A_N$ is Hurwitz, $C_N \mathcal{W}_N \nsubseteq B_N \mathcal{U}_N$, $\gamma \gamma_N < \alpha \alpha_N$ and $m \leq \frac{\sqrt{\lambda_{min}^{\hat{P}}}}{\|K\|}$, then
    \begin{equation}\label{eq: x_N bound underactuated}
        \|x_N(t)\|_{P_N} \leq \left\{\def\arraystretch{2}\begin{array}{l}
            \max\left\{ 0,\ e^{-\alpha_N t} \|x_N(0)\|_{P_N} + M(t) \right\} \qquad \text{if} \quad \|\chi(t)\|_{\hat{P}} > 0, \\
            \frac{z_{max}^{P_N}}{\alpha_N} + \left( \|x_N(0)\|_{P_N} - \frac{z_{max}^{P_N}}{\alpha_N} \right) e^{-\alpha_N t} \qquad \text{otherwise,}
        \end{array} \right.
    \end{equation}
    with 
    \begin{equation}\label{eq: M underactuated}
        M(t) = \frac{\alpha z_{max}^{P_N} \big(1 - e^{-\alpha_N t} \big) }{\alpha \alpha_N - \gamma \gamma_N} + e^{-\alpha_N t} \left( \frac{\gamma_N h_+}{r_+}\big( e^{r_+ t} -1 \big) + \frac{\gamma_N h_-}{r_-} \big( e^{r_- t} -1\big)  \right). 
    \end{equation}
\end{proposition}
\begin{proof}
    We recall from Proposition~\ref{prop: steady state xN} that $\|x_N(t)\|_{P_N} \leq e^{-\alpha_N t} \left( \|x_N(0)\|_{P_N} + \int_0^t e^{\alpha_N \tau} \beta_N(\tau)\, d\tau \right)$ \eqref{eq: x_N bound incomplete}.
    Following \eqref{eq: beta_N}, we have $\beta_N(t) \leq z_{max}^{P_N} + \gamma_N \|\chi(t)\|_{\hat{P}}$ where we bound $\|\chi(t)\|_{\hat{P}}$ with \eqref{eq: X bound not full rank}
    \begin{align*}
        \int_0^t e^{\alpha_N \tau} \beta_N(\tau)\, d\tau &\leq \int_0^t e^{\alpha_N \tau} \Big( z_{max}^{P_N} + \gamma_N p + \gamma_N h_+ e^{(r_+ - \alpha_N)\tau} + \gamma_N h_- e^{(r_- - \alpha_N)\tau} \Big) d\tau \\
        &= \frac{e^{\alpha_N t} -1}{\alpha_N}\big( z_{max}^{P_N} + \gamma_N p\big) + \frac{\gamma_N h_+}{r_+}\big( e^{r_+ t} -1 \big) + \frac{\gamma_N h_-}{r_-} \big( e^{r_- t} -1\big).
    \end{align*}
    We replace $p$ in
    \begin{equation*}
        \frac{z_{max}^{P_N} + \gamma_N p}{\alpha_N} = \frac{z_{max}^{P_N}(\alpha \alpha_N - \gamma \gamma_N) + \gamma_N \gamma z_{max}^{P_N} }{\alpha_N(\alpha \alpha_N - \gamma \gamma_N)} = \frac{\alpha z_{max}^{P_N} }{\alpha \alpha_N - \gamma \gamma_N} \quad \text{with} \quad p = \frac{\gamma z_{max}^{P_N}}{ \alpha \alpha_N - \gamma \gamma_N}.
    \end{equation*}
    Then, plugging the integral calculated above in \eqref{eq: x_N bound incomplete}, we obtain
    \begin{equation*}
        \|x_N(t)\|_{P_N} \hspace{-1mm} \leq e^{-\alpha_N t} \hspace{-1mm} \left( \hspace{-1mm} \|x_N(0)\|_{P_N} \hspace{-1mm} + \frac{\alpha z_{max}^{P_N} }{\alpha \alpha_N - \gamma \gamma_N}\big(e^{\alpha_N t} -1\big) + \frac{\gamma_N h_+}{r_+}\big( e^{r_+ t} -1 \big) + \frac{\gamma_N h_-}{r_-} \big( e^{r_- t} -1\big)  \hspace{-1mm} \right) \hspace{-1mm},
    \end{equation*}
   which yields \eqref{eq: x_N bound underactuated}.
    
    When $\|\chi(t)\|_{\hat{P}} = 0$, we have $\beta_N(t) \leq z_{max}^{P_N}$, which simplifies \eqref{eq: x_N bound incomplete} as follows
    \begin{align*}
        \|x_N(t)\|_{P_N} &\leq e^{-\alpha_N t} \left( \|x_N(0)\|_{P_N} + \int_0^t e^{\alpha_N \tau} z_{max}^{P_N}\, d\tau \right) = \frac{z_{max}^{P_N}}{\alpha_N} + \left( \|x_N(0)\|_{P_N} - \frac{z_{max}^{P_N}}{\alpha_N} \right) e^{-\alpha_N t}.
    \end{align*}
\end{proof}

Using Propositions~\ref{prop: steady state X underactuated} and \ref{prop: full steady state xN underactuated}, we can now quantify the effect of the loss of control authority over which the network was not resilient.
Without the full rank assumption on $\bar{B}$, we cannot stabilize the network, but we provide a guaranteed bound on its state. This constitutes our solution to Problem~\ref{prob: nonresilient stabilizability} for an underactuated network.

\section{Numerical examples}\label{sec: example}

We will now illustrate the theory established in the preceding sections on two academic examples, on an islanded microgrid \citep{bidram2013distributed, xie2019distributed, bidram2013secondary, guo2014distributed}, and on the IEEE 39-bus system \citep{IEEE_39}.
All the data and codes necessary to run the simulations in this section are available on GitHub\footnote{ \url{https://github.com/Jean-BaptisteBouvier/Network-Resilience}}
.

\subsection{Fully actuated 3-component network}\label{subsec: academic fully actuated}

We start by testing the results of Section~\ref{subsec: fully-actuated} on a simple network constituted of a nonresilient subsystem enduring a partial loss of control authority. This network of states $\chi_1$, $\chi_2$ and $x_N$ follows dynamics
\begin{equation}\label{eq:academic X}
    \dot \chi(t) = \begin{pmatrix} -1 & 0 \\ 0 & -1 \end{pmatrix} \chi(t) + \begin{pmatrix} 2 & 0 \\ 0 & 2 \end{pmatrix} \hat{u}(t) + \begin{pmatrix} 0 & 0.3 \\ 0.3 & 0 \end{pmatrix} \chi(t) + \begin{pmatrix} 0.3 \\ 0.3 \end{pmatrix} x_N(t), \quad \chi(0) = \begin{pmatrix} 1 \\ 1 \end{pmatrix},
\end{equation}
\begin{equation}\label{eq:academic x_N}
    \dot x_N(t) = -x_N(t) + u_N(t) + 2w_N(t) + \begin{pmatrix} 0.3 & 0.3 \end{pmatrix} \chi(t), \qquad x_N(0) = 0,
\end{equation}
with $\hat{u}(t) = \left(\begin{smallmatrix} \hat{u}_1(t) \\ \hat{u}_1(t) \end{smallmatrix}\right) \in [-1, 1]^2$, $u_N(t) \in \mathcal{U}_N = [-1, 1]$ and $w_N(t) \in \mathcal{W}_N = [-1, 1]$.
With the notation of \eqref{eq:split system N} and \eqref{eq: X N-1}, matrices $A_N$ and $\hat{A} + \hat{D}$ are both Hurwitz, and the control matrix $\hat{B}$ is full rank, with
\begin{equation*}
    A_N = -1, \qquad \hat{A} + \hat{D} = \begin{pmatrix} -1 & 0.3 \\ 0.3 & -1 \end{pmatrix}, \qquad \text{and} \qquad \hat{B} = \begin{pmatrix} 2 & 0 \\ 0 & 2 \end{pmatrix}.
\end{equation*}
Additionally, $C_N \mathcal{W}_N = [-2, 2] \nsubseteq B_N \mathcal{U}_N = [-1, 1]$. Thus, all the assumptions of Propositions~\ref{prop: steady state xN}, \ref{prop: steady state X} and \ref{prop: full steady state xN} are verified. To apply these results, we solve Lyapunov equations $A_N^\top P_N + P_N A_N = -Q_N$ and $(\hat{A}+\hat{D})^\top \hat{P} + \hat{P} (\hat{A}+\hat{D}) = -\hat{Q}$ with the function \textit{lyap} on MATLAB:
\begin{equation*}
    Q_N = 1, \quad P_N = 0.5, \qquad \hat{Q} = \begin{pmatrix} 1 & 0 \\ 0 & 1 \end{pmatrix}, \quad \hat{P} = \begin{pmatrix} 0.23 & 0.05 \\ 0.05 & 0.5 \end{pmatrix}.
\end{equation*}
Then, following Proposition~\ref{prop: steady state xN} $\alpha_N = 1$ and $z_{max}^{P_N} = 1$. From Proposition~\ref{prop: steady state X}, $b_{min}^{\hat{P}} = 2$, $\alpha = 0.7$, $\gamma = 0.51$, and $\gamma_N = 0.48$.

The stabilizability conditions of Theorem~\ref{thm: X stabilizable} are satisfied since $\gamma \gamma_N = 0.25 < \alpha \alpha_N = 0.7$ and $\gamma z_{max}^{P_N} = 0.5 < \alpha_N b_{min}^{\hat{P}} = 2$.
To verify that $\chi$ is indeed resiliently stabilizable in finite time by $\hat{B}\hat{u} = \frac{-\chi(t)}{\|\chi(t)\|_{\hat{P}}} b_{min}^{\hat{P}}$, we propagate $\chi(t)$ and $x_N(t)$ using control law $u_N(t) = -1$ and undesirable signal $w_N(t) = 1$. Figure~\ref{fig:academic full} shows the resulting states evolution along with the bounds of Propositions~\ref{prop: steady state xN}, \ref{prop: steady state X} and \ref{prop: full steady state xN}.

\begin{figure}[htb!]
    \centering
    \begin{subfigure}[]{0.49\textwidth}
        \includegraphics[scale=0.6]{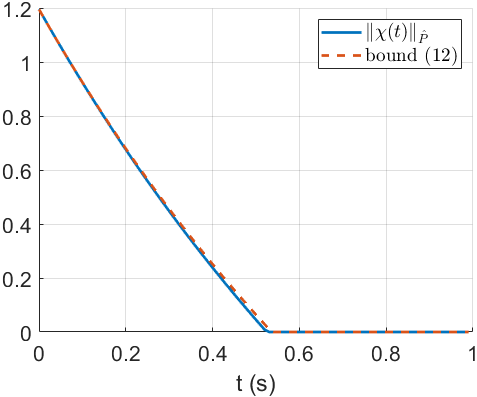}
        \caption{Finite-time resilient stabilization of network state $\chi(t)$ of \eqref{eq:academic X}.}
        \label{fig:academic full X}
    \end{subfigure}\hfill
    \begin{subfigure}[]{0.49\textwidth}
        \includegraphics[scale = 0.6]{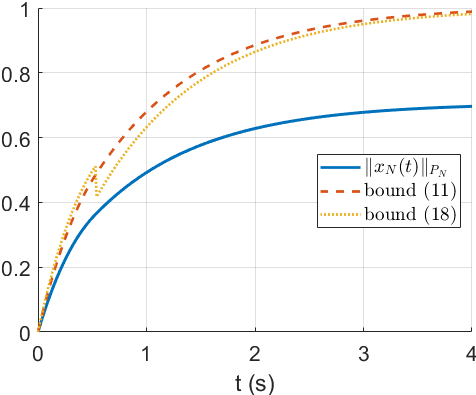}
        \caption{Resiliently bounded malfunctioning state $x_N(t)$ of \eqref{eq:academic x_N}.}
        \label{fig:academic full x_N}
    \end{subfigure}
    \caption{Time evolution of states $\chi$ and $x_N$ along with their bounds \eqref{eq: x_N bound incomplete}, \eqref{eq: X bound} and \eqref{eq: x_N bound}.}
    \label{fig:academic full}
\end{figure}

Figure~\ref{fig:academic full}(\subref{fig:academic full X}) shows the finite-time resilient stabilization of state $\chi$ and its respect of the tight bound~\eqref{eq: X bound}.
Figure~\ref{fig:academic full}(\subref{fig:academic full x_N}) illustrates the initial divergence but overall resilient boundedness of malfunctioning state $x_N$ while respecting both bounds \eqref{eq: x_N bound incomplete} and \eqref{eq: x_N bound}.
As discussed in Remark~\ref{rmk: discontinuous switch}, when $\chi(t)$ reaches $0$, bound \eqref{eq: x_N bound} operates a discontinuous switch.

\subsection{Underactuated 3-component network}\label{subsec: academic underactuated}

To validate the results of Section~\ref{subsec: nonresilient underactuated}, we need $\hat{B}$ not to be full row rank anymore, but the pair $(\hat{A} + \hat{D}, \hat{B})$ must remain controllable. Then, we remove the second column of $\hat{B}$ so that \eqref{eq:academic X} becomes
\begin{equation}\label{eq:academic X underactuated}
    \dot \chi(t) = \begin{pmatrix} -1 & 0.3 \\ 0.3 & -1 \end{pmatrix} \chi(t) + \begin{pmatrix} 2 \\ 0 \end{pmatrix} \hat{u}(t) + \begin{pmatrix} 0.3 \\ 0.3 \end{pmatrix} x_N(t), \qquad \chi(0) = \begin{pmatrix} 1 \\ 1 \end{pmatrix}.
\end{equation}
The MATLAB functions \textit{lqr} and \textit{lyap} choose the following gain matrix $K$ and positive definite matrices $\hat{P}$ and $\hat{Q}$:
\begin{equation*}
    K = \begin{pmatrix} 0.6383 & 0.1521 \end{pmatrix}, \qquad \hat{P} = \begin{pmatrix} 0.22 & 0.04 \\ 0.04 & 0.5 \end{pmatrix} \qquad \text{and} \qquad \hat{Q} = \begin{pmatrix} 1 & 0 \\ 0 & 1 \end{pmatrix}.
\end{equation*}
Then, $\gamma \gamma_N = 0.24 < \alpha \alpha_N = 0.98$ and $\hat{u}(t) := -K\chi(t) \in [-1, 1]$. Thus, the linear feedback of Proposition~\ref{prop: steady state X underactuated} is admissible and its bound \eqref{eq: X bound not full rank} holds as illustrated on Figure~\ref{fig:academic underactuated}(\subref{fig:academic X}).
We calculate bound \eqref{eq: X bound not full rank}, it has a minimum at $T = 1.9\, s$ and its supremum occurs either at $t = 0$ or as $t \rightarrow +\infty$. Figure~\ref{fig:academic underactuated}(\subref{fig:academic X}) validates the prediction and shows that the supremum of $b(t)$ is reached at $t = 0$.

\begin{figure}[htb!]
    \centering
    \begin{subfigure}[]{0.49\textwidth}
        \includegraphics[scale=0.6]{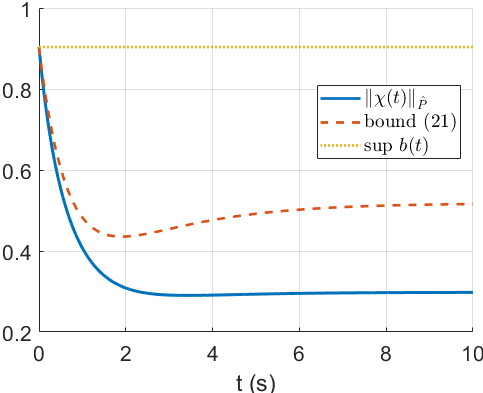}
        \caption{Resiliently bounded network state $\chi$ of \eqref{eq:academic X underactuated}.}
        \label{fig:academic X}
    \end{subfigure}\hfill
    \begin{subfigure}[]{0.49\textwidth}
        \includegraphics[scale = 0.6]{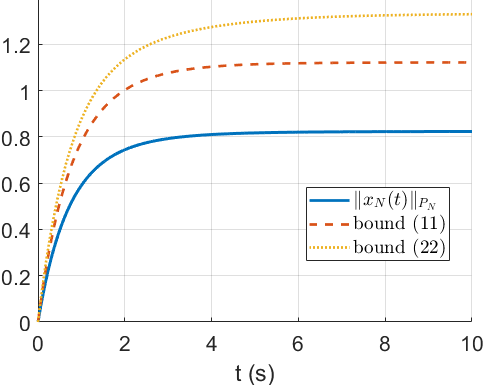}
        \caption{Resiliently bounded malfunctioning state $x_N$ of \eqref{eq:academic X underactuated}.}
        \label{fig:academic x_N}
    \end{subfigure}
    \caption{Illustration of bounds \eqref{eq: x_N bound incomplete}, \eqref{eq: X bound not full rank} and \eqref{eq: x_N bound underactuated} on the states $\chi$ and $x_N$.}
    \label{fig:academic underactuated}
\end{figure}

On the contrary to bound~\eqref{eq: x_N bound} on Figure~\ref{fig:academic full}(\subref{fig:academic full x_N}), bound~\eqref{eq: x_N bound underactuated} on Figure~\ref{fig:academic underactuated}(\subref{fig:academic x_N}) does not switch. Indeed, $\chi$ cannot be brought to $0$ by the linear control $\hat{u}$, as explained after Proposition~\ref{prop: steady state X underactuated}.

The bound of Proposition~\ref{prop: full steady state xN underactuated} is shown on Figure~\ref{fig:academic underactuated}(\subref{fig:academic x_N}) where no switch occurs because $\chi$ cannot be brought to $0$ by the linear control $\hat{u}$ as explained after Proposition~\ref{prop: steady state X underactuated}.

As illustrated on Figure~\ref{fig:academic underactuated}(\subref{fig:academic x_N}), bound~\eqref{eq: x_N bound incomplete} is tighter than bound~\eqref{eq: x_N bound underactuated}. The reason for this difference in conservatism is that \eqref{eq: x_N bound incomplete} uses directly the value of $\chi$, while \eqref{eq: x_N bound underactuated} replaces $\chi$ by its bound \eqref{eq: X bound not full rank}.

To verify the admissibility of the linear control law $\hat{u}(t) = -K \chi(t)$ we cannot use the sufficient condition of Proposition~\ref{prop: steady state X underactuated} as $\sup b(t) = 0.9 > \frac{\sqrt{\lambda_{min}^{\hat{P}}}}{\|K\|} = 0.71$. However, we can see on Figure~\ref{fig:academic KX} that $\|K\chi(t)\| \leq 1$ for all $t \geq 0$ and thus $\hat{u}$ is in fact admissible.
Note that $\hat{u}(t)$ does not converge to 0 since it needs to constantly counteract the destabilizing impact of $x_N(t)$ on $\chi(t)$.

\begin{figure}[htb!]
    \centering
    \includegraphics[scale=0.54]{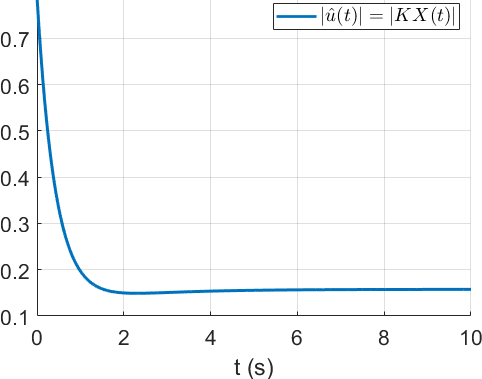}
    \caption{Linear feedback $\hat{u}(t) = -K\chi(t)$.}
    \label{fig:academic KX}
\end{figure}

\subsection{Microgrid test system}\label{subsec: microgrid}

We will now investigate the resilient stabilizability of the islanded microgrid illustrated on Figure~\ref{fig: microgrid} and studied in numerous power system works such as \citep{bidram2013distributed, xie2019distributed, guo2014distributed, bidram2013secondary}.

\begin{figure}[htb!]
    \centering
    \includegraphics[scale=0.3]{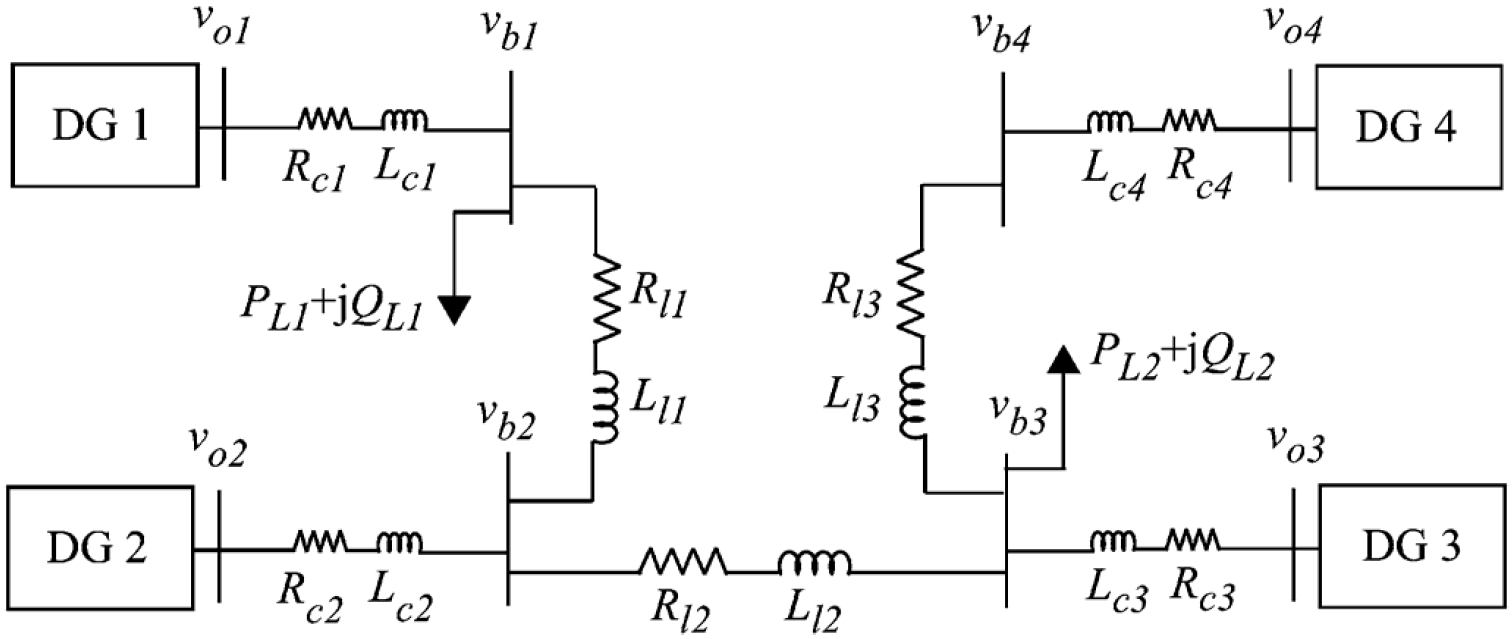}
    \caption{Single-line diagram of the microgrid test system from \citep{bidram2013distributed}.}
    \label{fig: microgrid}
\end{figure}

Distributed generator 1 (DG1) is connected to the leader node DG0 with pinning gain $1$ and all the DGs are connected following the communication digraph of Figure~\ref{fig: digraph}.

\begin{figure}[htbp!]
    \centering
    \begin{tikzpicture}[scale = 1]
        \node at (-2, 0.2) {reference};
        \node at (-2, -0.2) {value (DG0)};
        \draw[<->, very thick] (-1, 0) -- (-0.5, 0);
        
        \draw[] (0,0) circle (5mm) node[]{DG1};
        \draw[<->, very thick] (0.5, 0) -- (1, 0);
        \draw[] (1.5,0) circle (5mm) node[]{DG2};
        \draw[<->, very thick] (2, 0) -- (2.5, 0);
        \draw[] (3,0) circle (5mm) node[]{DG3};
        \draw[<->, very thick] (3.5, 0) -- (4, 0);
        \draw[] (4.5,0) circle (5mm) node[]{DG4};
        
    \end{tikzpicture}
    \caption{Topology of the communication digraph.}
    \label{fig: digraph}
\end{figure}
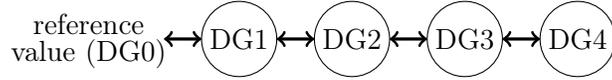

We follow works \citep{bidram2013distributed, xie2019distributed, bidram2013secondary} and employ their input-output feedback linearization of the DG dynamics. Since all DGs aim at synchronizing their voltage to the reference $v_0 = v_{ref}$, we consider as states the voltage difference between neighbors: $x_i := \big[v_{i} - v_{i-1}, \dot v_{i} - \dot v_{i-1}\big]^\top$ for $i \in [\![1,4]\!]$. Then, the objective is to stabilize all the $x_i$ to the origin. After a loss of control authority in DG4, we instead aim at bounding the voltages so that they do not diverge too far from the reference. The linearized microgrid is underactuated but controllable with $\gamma \gamma_q = 0.0399 < \alpha \alpha_q = 0.0401$, so that we can apply Propositions~\ref{prop: steady state X underactuated} and \ref{prop: full steady state xN underactuated}.

\begin{figure}[htbp!]
    \centering
    \includegraphics[scale=0.7]{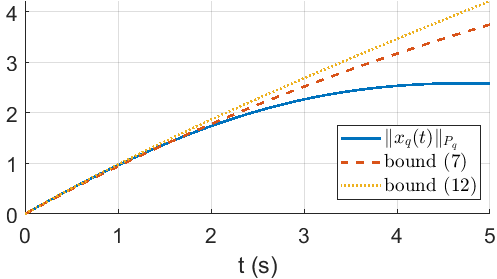}
    \caption{Bounded malfunctioning state $x_N = \big[v_4 - v_3, \dot v_4 - \dot v_3\big]^\top$ with tight bounds~\eqref{eq: x_N bound incomplete} and \eqref{eq: x_N bound underactuated}.}
    \label{fig: microgrid x_q}
\end{figure}

As seen on Figure~\ref{fig: microgrid x_q}, bounds~\eqref{eq: x_N bound incomplete} and \eqref{eq: x_N bound underactuated} are initially tight and only diverge slowly from $\|x_N(t)\|_{P_N}$.

\begin{figure}[htbp!]
    \centering
    \includegraphics[scale=0.7]{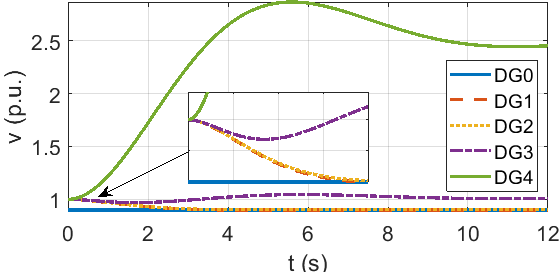}
    \caption{Magnitude of the output voltages of the reference and the four nodes of the microgrid.}
    \label{fig: microgrid V}
\end{figure}

Since DG4 is not resilient, $v_4$ does not converge to the reference $v_0$, but is nonetheless bounded as shown on Figure~\ref{fig: microgrid V}. In turn, DG4 disrupts its neighbor DG3, whose voltage $v_3$ cannot reach $v_0$ either. However, $v_3$ is maintained much closer to $v_0$ than $v_4$ thanks to the resilient controller of Proposition~\ref{prop: steady state X underactuated}. This resilient controller also allows DG1 and DG2 to remain completely oblivious of the loss of control of DG4.

\subsection{Resilient stabilizability of the IEEE 39-bus system}\label{subsec: IEEE 39}

In this section, we will illustrate the results of Section~\ref{subsec: nonresilient underactuated} on the IEEE 39-bus system \citep{IEEE_39} linearized in \citep{Sai_IEEE_39} and represented on Figure~\ref{fig:IEEE 39}.

\begin{figure}[htb!]
    \centering
    \includegraphics[scale=0.8]{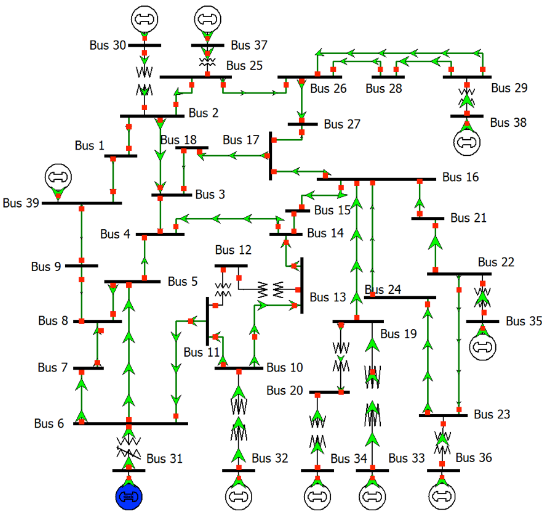}
    \caption{Illustration of the IEEE 39-bus system \citep{IEEE_39} obtained from \url{https://icseg.iti.illinois.edu/ieee-39-bus-system/}.}
    \label{fig:IEEE 39}
\end{figure}

This system is comprised of 29 load buses (1 to 29 on Figure~\ref{fig:IEEE 39}) and 10 generator buses (30 to 39 on Figure~\ref{fig:IEEE 39}). The state of the load buses is described solely by their phase angles $\{\delta_i\}_{i \, \in\, [\![1,29]\!]}$, while the state of the generators is composed of phase angles $\delta_{i \, \in\, [\![30,39]\!]}$ and frequencies $\omega_{i \, \in\, [\![30,39]\!]}$, which leads to 49 states. Only the generator buses possess a control input $u_g$.
Following \citep{Sai_IEEE_39}, the power network equations can be linearized around their nominal operating point after adjustment for the reference bus, chosen to be the first generator, i.e., bus 30. The state vector is then
\begin{equation*}
    x = \big( \{ \delta_i - \delta_{30}\}_{i\, \in\, [\![1,29]\!] \cup [\![31,39]\!]}, \{\omega_i\}_{i \, \in\, [\![30,39]\!]} \big) \in \mathbb{R}^{48}.
\end{equation*}


After a cyber-attack, the network controller loses control authority over the actuator of generator bus 39, i.e., $x_N = \big( \delta_{39} \ \omega_{39} \big)^\top$ and $w_N = u_{39}$.
Following \citep{Sai_IEEE_39}, the malfunctioning dynamics are then
\begin{equation}\label{eq: malfunctioning 39}
    \begin{pmatrix} \dot \delta_{39}(t) \\ \dot \omega_{39}(t) \end{pmatrix} = \begin{pmatrix} 0 & 1 \\ -18.63 & -11.22 \end{pmatrix} \begin{pmatrix} \delta_{39}(t) \\ \omega_{39}(t) \end{pmatrix} + \begin{pmatrix} 0 \\ 0.222 \end{pmatrix} w_N(t) + D_{N,-} \chi(t).
\end{equation}
As in Section~\ref{subsec: academic underactuated}, we choose the initial states to be $\chi(0) = \mathbf{1}_{46}$ and $x_N(0) = (0 0)^\top$. Since $A_N$ is Hurwitz and $B_N = 0$, the assumptions of Proposition~\ref{prop: steady state xN} are satisfied.
Additionally, pair $(\hat{A} + \hat{D}, \hat{B})$ is controllable so we can find a stabilizing gain matrix $K$ for the network dynamics. However, we cannot apply Proposition~\ref{prop: steady state X underactuated} because the stability condition $\gamma \gamma_N < \alpha \alpha_N$ is not satisfied. Indeed, $\gamma \gamma_N = 6.3 \times 10^4$, while $\alpha \alpha_N = 5.7 \times 10^{-3}$. 
This magnitude difference leads to the exponential divergence of bound~\eqref{eq: X bound not full rank} and of bound~\eqref{eq: x_N bound underactuated}, as seen on Figure~\ref{fig:IEEE}(\subref{fig:IEEE X}) and \ref{fig:IEEE}(\subref{fig:IEEE x_N}), respectively.

\begin{figure}[htb!]
    \centering
    \begin{subfigure}[]{0.49\textwidth}
        \includegraphics[scale=0.6]{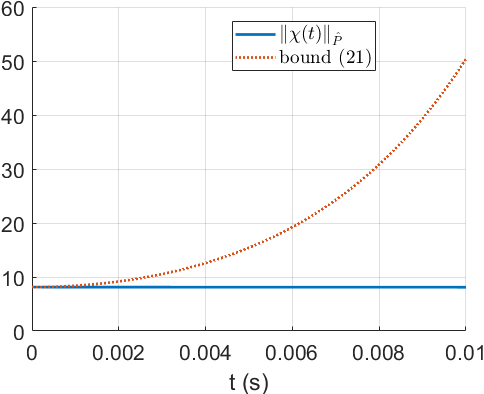}
        \caption{Simulation of network state $\chi$ of the IEEE 39-bus system with exponentially diverging bound~\eqref{eq: X bound not full rank}.}
        \label{fig:IEEE X}
    \end{subfigure}\hfill
    \begin{subfigure}[]{0.49\textwidth}
        \includegraphics[scale = 0.6]{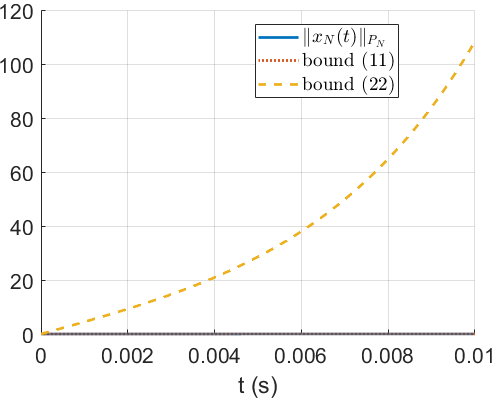}
        \caption{Simulation of malfunctioning state $x_N$ of the IEEE 39-bus system with tight bound~\eqref{eq: x_N bound incomplete} and exponentially diverging bound~\eqref{eq: x_N bound underactuated}.}
        \label{fig:IEEE x_N}
    \end{subfigure}
    \caption{Illustration of bounds \eqref{eq: x_N bound incomplete}, \eqref{eq: X bound not full rank} and \eqref{eq: x_N bound underactuated} on the states $\chi$ and $x_N$.}
    \label{fig:IEEE}
\end{figure}

Note that bound~\eqref{eq: x_N bound incomplete} is much tighter than \eqref{eq: x_N bound underactuated} on Figure~\ref{fig:IEEE}(\subref{fig:IEEE x_N}) because bound~\eqref{eq: x_N bound incomplete} uses 
$$\int_0^t e^{-\alpha_N(t-\tau)}\|D_{N,\_} \chi(\tau)\|_{P_N}d\tau,$$ whereas \eqref{eq: x_N bound underactuated} bounds this integral with exponentially diverging \eqref{eq: X bound not full rank}.
In fact, bound~\eqref{eq: x_N bound incomplete} remains a reasonable bound for malfunctioning state $x_N$ over a much longer time horizon as illustrated on Figure~\ref{fig:IEEE x_N long}.

\begin{figure}[htb!]
    \centering
    \includegraphics[scale = 0.6]{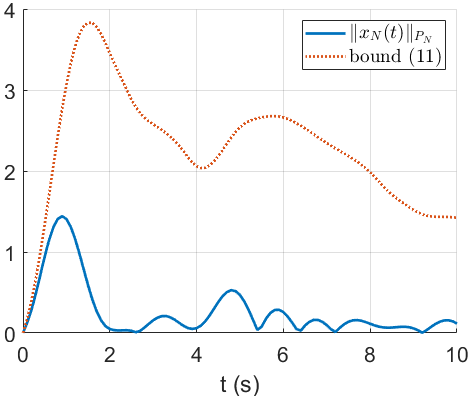}
    \caption{Simulation of malfunctioning state $x_N$ of the IEEE 39-bus system with tight bound~\eqref{eq: x_N bound incomplete}.}
    \label{fig:IEEE x_N long}
\end{figure}

As before, sufficient condition $\underset{t\, \geq\, 0}{\sup}\ b(t) \leq \frac{\sqrt{\lambda_{min}^{\hat{P}}}}{\|K\|}$ of Proposition~\ref{prop: steady state X underactuated} cannot tell whether linear feedback $\hat{u}$ is admissible. However, the choice of $K$ ensures admissibility $\underset{i, t}{\max} |K\chi_i(t)| \leq 1$ as shown on Figure~\ref{fig:IEEE KX}.

\begin{figure}[htb!]
    \centering
    \includegraphics[scale=0.6]{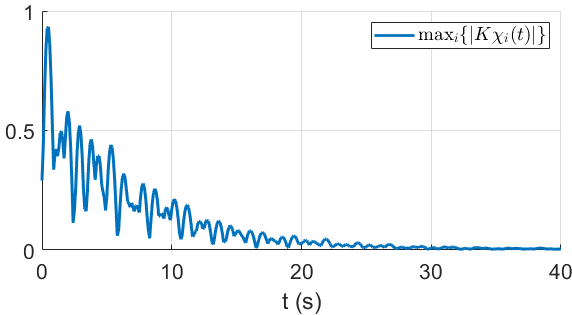}
    \caption{Maximal component of the linear feedback $\hat{u}(t) = -K\chi(t)$.}
    \label{fig:IEEE KX}
\end{figure}

Let us delve a bit deeper into the exponential divergence of bound~\eqref{eq: X bound not full rank}. As mentioned previously, bound~\eqref{eq: X bound not full rank} is not tight because $\gamma \gamma_N = 6.3 \times 10^4$ is orders of magnitude larger than $\alpha \alpha_N = 5.7 \times 10^{-3}$, whereas the stability condition of Proposition~\ref{prop: steady state X underactuated} calls for $\gamma \gamma_N < \alpha \alpha_N$. As discussed after Theorem~\ref{thm: X stabilizable}, this condition carries the intuition that the perturbations arising from the coupling between $x_N$ and $\chi$ should be weaker than their respective stability.
Despite having $\gamma \gamma_N >> \alpha \alpha_N$, the coupling does not destabilize states $x_N$ and $\chi$, which are both bounded, as shown on Figure~\ref{fig:IEEE X x_N}.

\begin{figure}[htb!]
    \centering
    \includegraphics[scale = 0.6]{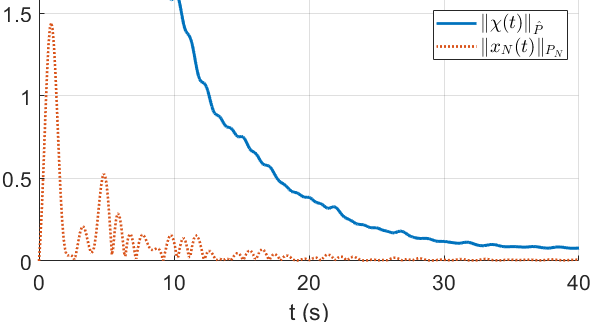}
    \caption{Simulation of network state $\chi$ and malfunctioning state $x_N$ of the IEEE 39-bus system.}
    \label{fig:IEEE X x_N}
\end{figure}

Since the coupling does not destabilize states $\chi$ and $x_N$, the violation of stability condition $\gamma \gamma_N < \alpha \alpha_N$ is in fact due to the failure of parameters $\gamma$ and $\gamma_N$ to characterize the coupling between states $\chi$ and $x_N$.
As shown on Figure~\ref{fig:IEEE 39} each bus is only connected to a small number of other buses. Then, matrix $\hat{D}$ is almost entirely composed of zeros except for a handful of terms per row. Because of this strong coupling with very few nodes, constants $\gamma$ and $\gamma_N$ are very large. However, the sparsity of matrix $\hat{D}$ results in weak coupling of states $\chi$ and $x_N$, rendering $\gamma$ and $\gamma_N$ overly conservative.
This intuition was illustrated on the more densely connected microgrid of Section~\ref{subsec: microgrid}, where the bounds were much tighter. To study sparsely connected networks like the IEEE 39-bus system, we have the intuition that choosing a different norm reflecting the sparsity of matrix $\hat{D}$ would lower the values of $\gamma$ and $\gamma_N$. Doing so would significantly and non-trivially alter all the proofs of Section~\ref{sec: nonresilient}.

\section{Conclusion}\label{sec:conclusion}

This paper investigated the resilient stabilizability of linear networks enduring a loss of control.
We first saw that the overall stabilizability of networks composed exclusively of resilient subsystems depends only on their interconnection.
Then, we focused on networks losing control authority over a nonresilient subsystem.
In this scenario, we showed that under some conditions, the state of underactuated networks can remain bounded and the state of fully actuated networks can be stabilized.
We were able to quantify the maximal magnitude of undesirable inputs that can be applied to a nonresilient subsystem without destabilizing the rest of the network.

We are considering several avenues of future work.
First, building on the nonlinear resilience theory of \citep{JGCD}, we would like to extend our approach to nonlinear networks. Doing so would allow us to study the true nonlinear dynamics of power systems, including the IEEE 39-bus system.
Second, following the discussion at the end of Section~\ref{subsec: IEEE 39} we want to extend this theory to different matrix norms to provide tighter bounds for sparse coupling matrices.
The last avenue of future work would be to relax the assumption of real-time knowledge of the undesirable inputs by the controller. Doing so would allow to account for actuation delays and can possibly be accomplished following the techniques introduced in \citep{JGCD}.

\appendix
\section{Supporting lemmata}\label{apx}

\begin{lemma}\label{lemma: multi losses}
    A loss of control authority over actuators of several subsystems of network~\eqref{eq:N network D} is equivalent to the loss of control authority over a new single subsystem composed of the `stacked' malfunctioning subsystems.
\end{lemma}
\begin{proof}
    If multiple subsystems simultaneously lose control authority over some of their actuators, we can take advantage of the linearity of the dynamics of network~\eqref{eq:N network D} to combine these malfunctioning subsystems into a single macro subsystem whose dynamics have the same form as a single malfunctioning subsystem. 
    We will describe this process when two subsystems simultaneously lose control authority over some of their actuators. Note that the same process works for more than two subsystems, the notations will just become more cumbersome. In the two subsystem case, assume that subsystems $i$ and $j$ are attacked. Then, we can write their dynamics as follows:
    \begin{equation}\label{eq: subsys split}
        \begin{array}{c}
            \dot x_i(t) = A_i x_i(t) + B_i u_i(t) + C_i w_i(t) + \displaystyle\sum_{k\, \in\, \mathcal{N}_i} D_{i,k} x_k(t), \\ 
            \dot x_j(t) = A_j x_j(t) + B_j u_j(t) + C_j w_j(t) + \displaystyle\sum_{k\, \in\, \mathcal{N}_j} D_{j,k} x_k(t).
        \end{array} 
    \end{equation}   
    Let us now define the combined vectors and matrices distinguishable with their subscript $_c$ as follows:
    \begin{equation*}
        \hat{x}_c(t) := \begin{pmatrix} x_i(t) \\ x_j(t) \end{pmatrix}, \qquad \hat{u}_c(t) := \begin{pmatrix} u_i(t) \\ u_j(t) \end{pmatrix}, \qquad \hat{w}_c(t) := \begin{pmatrix} w_i(t) \\ w_j(t) \end{pmatrix},       
    \end{equation*}
    \begin{equation*}
        \hat{A}_c := \begin{pmatrix} A_i & 0 \\ 0 & A_j \end{pmatrix}, \qquad \hat{B}_c := \begin{pmatrix} B_i & 0 \\ 0 & B_j \end{pmatrix}, \qquad \hat{C}_c := \begin{pmatrix} C_i & 0 \\ 0 & C_j \end{pmatrix}.
    \end{equation*}
    We also define the set of combined neighbors as $\hat{\mathcal{N}}_c := \mathcal{N}_i \cup \mathcal{N}_j$. We partition this set as the union of the following disjoint sets: 
    \begin{equation*}
        \hat{\mathcal{N}}_c = (\mathcal{N}_i \cap \mathcal{N}_j)\ \sqcup\ (\mathcal{N}_i - \mathcal{N}_j)\ \sqcup\ (\mathcal{N}_j - \mathcal{N}_i).
    \end{equation*}
    Then, for $k \in \hat{\mathcal{N}}_c$ we define the combined coupling matrix $\hat{D}_{c,k}$ depending on which subset of the partition of $\hat{\mathcal{N}}_c$ the index $k$ belongs to
    \begin{align*}
        \hat{D}_{c,k} &:= \begin{pmatrix} D_{i,k} \\ D_{j,k} \end{pmatrix} \qquad \text{if}\ k \in \mathcal{N}_i \cap \mathcal{N}_j, \\
        \hat{D}_{c,k} &:= \begin{pmatrix} D_{i,k} \\ 0 \end{pmatrix} \qquad \text{if}\ k \in \mathcal{N}_i - \mathcal{N}_j, \\
        \hat{D}_{c,k} &:= \begin{pmatrix} 0 \\ D_{j,k} \end{pmatrix} \qquad \text{if}\ k \in \mathcal{N}_j - \mathcal{N}_i.
    \end{align*}
    With the partition of $\hat{\mathcal{N}}_c$, we have
    \begin{align*}
        \sum_{k\, \in\, \hat{\mathcal{N}}_c} \hat{D}_{c,k} x_k(t) &= \sum_{k\, \in\, \mathcal{N}_i \cap \mathcal{N}_j} \hat{D}_{c,k} x_k(t) + \sum_{k\, \in\, \mathcal{N}_i - \mathcal{N}_j} \hat{D}_{c,k} x_k(t) + \sum_{k\, \in\, \mathcal{N}_j - \mathcal{N}_i} \hat{D}_{c,k} x_k(t) \\
        &= \sum_{k\, \in\, \mathcal{N}_i \cap \mathcal{N}_j} \begin{pmatrix} D_{i,k} \\ D_{j,k} \end{pmatrix} x_k(t) + \sum_{k\, \in\, \mathcal{N}_i - \mathcal{N}_j} \begin{pmatrix} D_{i,k} \\ 0 \end{pmatrix} x_k(t) + \sum_{k\, \in\, \mathcal{N}_j - \mathcal{N}_i} \begin{pmatrix} 0 \\ D_{j,k} \end{pmatrix} x_k(t) \\
        &= \sum_{k\, \in\, \mathcal{N}_i \cap \mathcal{N}_j} \begin{pmatrix} D_{i,k} x_k(t) \\ D_{j,k} x_k(t) \end{pmatrix} + \sum_{k\, \in\, \mathcal{N}_i - \mathcal{N}_j} \begin{pmatrix} D_{i,k}  x_k(t)\\ 0 \end{pmatrix} + \sum_{k\, \in\, \mathcal{N}_j - \mathcal{N}_i} \begin{pmatrix} 0 \\ D_{j,k}  x_k(t)\end{pmatrix} \\
        &= \begin{pmatrix}
            \displaystyle\sum_{k\, \in\, \mathcal{N}_i \cap \mathcal{N}_j} D_{i,k} x_k(t) + \displaystyle\sum_{k\, \in\, \mathcal{N}_i - \mathcal{N}_j} D_{i,k}  x_k(t) + \displaystyle\sum_{k\, \in\, \mathcal{N}_j - \mathcal{N}_i} 0 \\
            \displaystyle\sum_{k\, \in\, \mathcal{N}_i \cap \mathcal{N}_j} D_{j,k} x_k(t) + \displaystyle\sum_{k\, \in\, \mathcal{N}_i - \mathcal{N}_j} 0 + \displaystyle\sum_{k\, \in\, \mathcal{N}_j - \mathcal{N}_i} D_{j,k}  x_k(t)
        \end{pmatrix} = \begin{pmatrix}
            \displaystyle\sum_{k\, \in\, \mathcal{N}_i} D_{i,k} x_k(t) \\
            \displaystyle\sum_{k\, \in\, \mathcal{N}_j} D_{j,k} x_k(t) 
        \end{pmatrix},
    \end{align*}
    since $(\mathcal{N}_i \cap \mathcal{N}_j) \sqcup (\mathcal{N}_i - \mathcal{N}_j) = \mathcal{N}_i$ and $(\mathcal{N}_i \cap \mathcal{N}_j) \sqcup (\mathcal{N}_j - \mathcal{N}_i) = \mathcal{N}_j$.
    Then, \eqref{eq: subsys split} can be written as
    \begin{align}
        \begin{pmatrix} \dot x_i(t) \\ \dot x_j(t) \end{pmatrix} &= \begin{pmatrix}
        A_i x_i(t) + B_i u_i(t) + C_i w_i(t) + \displaystyle\sum_{k\, \in\, \mathcal{N}_i} D_{i,k} x_k(t)\\
        A_j x_j(t) + B_j u_j(t) + C_j w_j(t) + \displaystyle\sum_{k\, \in\, \mathcal{N}_j} D_{j,k} x_k(t)
        \end{pmatrix} = \dot{\hat{x}}_c(t) \nonumber \\
        &= \begin{pmatrix} A_i & 0 \\ 0 & A_j \end{pmatrix} \begin{pmatrix} x_i(t) \\ x_j(t) \end{pmatrix} + \begin{pmatrix} B_i & 0 \\ 0 & B_j \end{pmatrix} \begin{pmatrix} u_i(t) \\ u_j(t) \end{pmatrix} + \begin{pmatrix} C_i & 0 \\ 0 & C_j \end{pmatrix} \begin{pmatrix} w_i(t) \\ w_j(t) \end{pmatrix} + \begin{pmatrix}
            \displaystyle\sum_{k\, \in\, \mathcal{N}_i} D_{i,k} x_k(t) \\
            \displaystyle\sum_{k\, \in\, \mathcal{N}_j} D_{j,k} x_k(t) 
        \end{pmatrix} \nonumber \\
        &= \hat{A}_c \hat{x}_c(t) + \hat{B}_c \hat{u}_c(t) + \hat{C}_c \hat{w}_c(t) + \sum_{k\, \in\, \hat{\mathcal{N}}_c} \hat{D}_{c,k} x_k(t) = \dot{\hat{x}}_c(t). \label{eq: combined}
    \end{align}
    The combination of the two malfunctioning subsystems of \eqref{eq: subsys split} into a single one \eqref{eq: combined} allows us to study the loss of a single subsystem without losing generality in the number of malfunctioning subsystems. 
\end{proof}

\begin{lemma}\label{lemma: controllability radius}
    $\mu_{\bar{B}}(A) \geq \mu(A, \bar{B})$.
\end{lemma}
\begin{proof}
    Define sets
    \begin{align*}
        \mathcal{S}_A &:= \big\{ ( \Delta A, 0) \in \mathbb{R}^{n \times n} \times \mathbb{R}^{n \times m} : (A+ \Delta A, \bar{B} + 0)\ \text{is uncontrollable}\big\}, \\
        \mathcal{S}_{AB} &:= \big\{ ( \Delta A, \Delta \bar{B}) \in \mathbb{R}^{n \times n} \times \mathbb{R}^{n \times m} : (A+ \Delta A, \bar{B} + \Delta \bar{B})\ \text{is uncontrollable}\big\}.
    \end{align*}
    Notice that $\mathcal{S}_A \subseteq \mathcal{S}_{AB}$. Therefore,
    \begin{equation*}
        \min \big\{ \| \Delta A, \Delta \bar{B}\| : (\Delta A, \Delta \bar{B}) \in \mathcal{S}_A \big\} \geq \min \big\{ \| \Delta A, \Delta \bar{B}\| : (\Delta A, \Delta \bar{B}) \in \mathcal{S}_{AB} \big\},
    \end{equation*}
    i.e., $\mu_{\bar{B}}(A) \geq \mu(A, \bar{B})$.
\end{proof}

The following result relates the norms induced by two positive definite matrices of different sizes.

\begin{lemma}\label{lemma: P norm change}
    Let $D \in \mathbb{R}^{m \times n}$, $P \in \mathbb{R}^{m \times m}$ and $Q \in \mathbb{R}^{n \times n}$. If $P = P^\top \succ 0$ and $Q \succ 0$, then $\|Dx\|_P \leq \|x\|_Q \sqrt{\frac{\lambda_{max}^{D^\top P D}}{\lambda_{min}^Q}}$ for all $x \in \mathbb{R}^n$.
\end{lemma}
\begin{proof}
    Applying the Rayleigh quotient inequality \citep{matrix_computations} to symmetric matrices $Q$ and $D^\top P D$ yields, 
    \begin{equation*}
        \lambda_{min}^Q \leq \frac{x^\top Q x}{x^\top x} \quad \text{and} \quad \frac{x^\top D^\top P D x}{x^\top x} \leq \lambda_{max}^{D^\top P D},
    \end{equation*}
    for all $x \in \mathbb{R}^n$, $x\neq 0$. Then,
    \begin{align*}
        \|Dx\|_P &= \sqrt{x^\top D^\top P D x} \leq \sqrt{\lambda_{max}^{D^\top P D} } \sqrt{x^\top x} \\
        &\leq \sqrt{ \frac{\lambda_{max}^{D^\top P D}}{\lambda_{min}^Q} } \sqrt{x^\top Q x} = \|x\|_Q \sqrt{ \frac{\lambda_{max}^{D^\top P D}}{\lambda_{min}^Q} }.
    \end{align*}
\end{proof}

Since $(x,y) \mapsto x^\top P y$ defines a scalar product for any $P \succ 0$, it verifies the Cauchy-Schwarz inequality \citep{matrix_computations}. We provide here a more constructive proof of this result for the reader.

\begin{lemma}[Cauchy-Schwarz inequality for the $P$-norm]\label{lemma: CS}
    Let $P \in \mathbb{R}^{n \times n}$, $P \succ 0$ and $x \in \mathbb{R}^n$, $y \in \mathbb{R}^n$. Then, $x^\top P y \leq \|x\|_P \|y\|_P$.
\end{lemma}
\begin{proof}
    Since $P \succ 0$, there exists a matrix $M \in \mathbb{R}^{n \times n}$ such that $P = M^\top M$ \citep{matrix_computations}. Then,
    \begin{equation*}
        x^\top P y = x^\top M^\top M y = (Mx)^\top My \leq \|Mx\| \ \|My\|,
    \end{equation*}
    by the Cauchy-Schwarz inequality applied to the Euclidean norm on $\mathbb{R}^n$ \citep{matrix_computations}. Note that
    \begin{equation*}
        \|Mx\| = \sqrt{ (Mx)^\top Mx} = \sqrt{x^\top M^\top M x} = \sqrt{x^\top Px} = \|x\|_P.
    \end{equation*}
    Similarly, $\|My\| = \|y\|_P$. Thus, $x^\top P y \leq \|x\|_P \|y\|_P$.
\end{proof}

We now show how the non-resilience of subsytem~\eqref{eq:split system N} translates to a positive $z_{max}^{P_N}$.

\begin{lemma}\label{lemma: z_max > 0}
    With $P_N \succ 0$ and $z_{max}^{P_N} = \underset{w_N\, \in\, \mathcal{W}_N}{\max} \big\{ \underset{u_N\, \in\, \mathcal{U}_N}{\min} \|C_N w_N + B_N u_N\|_{P_N} \big\}$, we have $-C_N \mathcal{W}_N \nsubseteq B_N \mathcal{U}_N \iff z_{max}^{P_N} > 0$.
\end{lemma}
\begin{proof}
    If $-C_N \mathcal{W}_N \subseteq B_N \mathcal{U}_N$, then for all $w_N \in \mathcal{W}_N$, there exists $u_N \in \mathcal{U}_N$ such that $C_N w_N + B_N u_N = 0$. Hence, $\underset{u_N\, \in\, \mathcal{U}_N}{\min} \big\{ \|C_N w_N + B_N u_N\|_{P_N} \big\} = 0$ for all $w_N \in \mathcal{W}_N$, i.e., $z_{max}^{P_N} = 0$.
    
    On the other hand, if $-C_N \mathcal{W}_N \nsubseteq B_N \mathcal{U}_N$, there exists $w_N \in \mathcal{W}_N$ such that $C_N w_N + B_N u_N \neq 0$ for all $u_N \in \mathcal{U}_N$. The function $u_N \mapsto \|C_N w_N + B_N u_N\|_{P_N}$ is continuous, nonnegative and $\mathcal{U}_N$ is compact, hence it reaches a minimum which cannot be null on $\mathcal{U}_N$, i.e., $\underset{u_N\, \in\, \mathcal{U}_N}{\min} \big\{ \|C_N w_N + B_N u_N\|_{P_N} \big\} > 0$. Then, $z_{max}^{P_N} > 0$.
\end{proof}

\begin{lemma}\label{lemma: ODE calculations}
    Detailed calculations for the proof of Proposition~\ref{prop: steady state X}.
\end{lemma}
\begin{proof}
    We first study the linear homogeneous differential equation associated with \eqref{eq: non-homogeneous ODE}, i.e.,
    \begin{equation}\label{eq: homogeneous ODE}
        \ddot s(t) + (\alpha - \alpha_N) \dot s(t) - \gamma \gamma_N s(t) = 0.
    \end{equation}
    Solutions of \eqref{eq: homogeneous ODE} can be written as $s_h(t) = e^{rt}$ with $r \in \mathbb{C}$. Plugging $s_h$ in \eqref{eq: homogeneous ODE} leads to the quadratic equation $r^2 + (\alpha - \alpha_N) r - \gamma \gamma_N = 0$ after diving by $e^{rt}$. The solutions of this quadratic equation are $r_{\pm} = \frac{1}{2} \big( \alpha_N - \alpha \pm \sqrt{(\alpha - \alpha_N)^2 + 4\gamma \gamma_N} \big)$. Notice that the discriminant is nonnegative, since $\gamma \geq 0$ and $\gamma_N \geq 0$, so both $r_\pm \in \mathbb{R}$.
    We also need a particular solution of the non-homogeneous equation \eqref{eq: non-homogeneous ODE}. We take $p \in \mathbb{R}$ such that $s_p(t) = p e^{\alpha_N t}$ and plug it in \eqref{eq: non-homogeneous ODE} to obtain
    \begin{align*}
        & \big( p \alpha_N^2 + (\alpha - \alpha_N) p \alpha_N - \gamma \gamma_N p - \gamma z_{max}^{P_N} + \alpha_N b_{min}^{\hat{P}} \big)e^{\alpha_N t} = 0, \\
        \text{i.e.,} \quad &p = \frac{\gamma z_{max}^{P_N} - \alpha_N b_{min}^{\hat{P}}}{ \alpha_N^2 + (\alpha - \alpha_N) \alpha_N - \gamma \gamma_N} = \frac{\gamma z_{max}^{P_N} - \alpha_N b_{min}^{\hat{P}}}{ \alpha \alpha_N - \gamma \gamma_N}.
    \end{align*}
    Let us first treat the case where $\alpha \alpha_N \neq \gamma \gamma_N$, so that $p$ is well-defined. In this case, the general solution of \eqref{eq: non-homogeneous ODE} is $s(t) = p e^{\alpha_N t} + h_+ e^{r_+ t} + h_- e^{r_- t}$ with $h_\pm \in \mathbb{R}$ two constants to choose. Since we obtained our solution by solving $\ddot s(t) = \frac{d f}{dt}\big(t, s(t) \big)$ instead of $\dot s(t) = f\big( t, s(t) \big)$, we have an additional initial condition to satisfy: $\dot s(0) = f\big(0, s(0) \big)$.
    
    Now we can apply the Comparison Lemma of \citep{Khalil} stating that if $\dot s(t) = f\big( t, s(t) \big)$, $f$ is continuous in $t$ and locally Lipschitz in $s$ and $s(0) = v(0)$, then $\dot v(t) \leq f\big( t, v(t) \big)$ implies $v(t) \leq s(t)$ for all $t \geq 0$. Using $\|\chi(t)\|_{\hat{P}} = y(t) = e^{-\alpha_N t} v(t) \leq e^{-\alpha_N t} s(t)$, we finally obtain \eqref{eq: X bound}.
    To determine the value of the constants $h_\pm$, we use the initial conditions $s(0) = v(0) = y(0)$ and $\dot s(0) = f \big( 0, s(0) \big)$, i.e., 
    \begin{equation*}
        p + h_+ + h_- = \|\chi(0)\|_{\hat{P}} \quad \text{and} \quad \alpha_N p + h_+ r_+ + h_- r_- = (\alpha_N - \alpha) \|\chi(0)\|_{\hat{P}} - b_{min}^{\hat{P}} + \gamma \|x_N(0)\|_{P_N}.
    \end{equation*}
    We can solve these equations as
    \begin{equation*}
        h_{\pm} = \frac{(\alpha_N - \alpha - r_{\mp})\|\chi(0)\|_{\hat{P}} + \gamma \|x_N(0)\|_{P_N} -b_{min}^{\hat{P}} + (r_{\mp} - \alpha_N)p}{\pm\sqrt{(\alpha - \alpha_N)^2 + 4\gamma \gamma_N}}.
    \end{equation*}
    
    In the case $\alpha \alpha_N = \gamma \gamma_N$, the discriminant of the quadratic equation arising from the homogeneous differential equation is $(\alpha - \alpha_N)^2 + 4\alpha \alpha_N = (\alpha + \alpha_N)^2$, which yields $r_+ = \alpha_N$ and $r_- = -\alpha$. Hence $e^{\alpha_N t}$ is an homogeneous solution and cannot be a particular solution of the non-homogeneous differential equation \eqref{eq: non-homogeneous ODE}. Instead, we try $s_p(t) = p t e^{\alpha_N t}$ as a particular solution. We calculate its derivatives $\dot s_p(t) = p(1 + \alpha_N t) e^{\alpha_N t}$, $\ddot s_p(t) = p(2\alpha_N + \alpha_N^2 t) e^{\alpha_N t}$ and plug it in \eqref{eq: non-homogeneous ODE}. After dividing by $e^{\alpha_N t}$ we obtain
    \begin{align*}
         0 &= p(2\alpha_N + \alpha_N^2 t) + (\alpha - \alpha_N)p(1 + \alpha_N t) - \alpha \alpha_N pt - \gamma z_{max}^{P_N} + \alpha_N b_{min}^{\hat{P}} \\
         &= p( 2\alpha_N + \alpha - \alpha_N) + pt\big( \alpha_N^2 +\alpha_N (\alpha - \alpha_N) - \alpha \alpha_N \big) -  \gamma z_{max}^{P_N} + \alpha_N b_{min}^{\hat{P}},
    \end{align*}
    i.e., $p = \frac{\gamma z_{max}^{P_N} - \alpha_N b_{min}^{\hat{P}}}{\alpha + \alpha_N}$. In this case $p$ is well-defined since $\alpha > 0$ and $\alpha_N > 0$. The general solution is then $s(t) = pt e^{\alpha_N t} + h_+ e^{\alpha_N t} + h_- e^{-\alpha t}$ with $h_\pm \in \mathbb{R}$ two constants. Applying the Comparison Lemma of \citep{Khalil} as above, we obtain $\|\chi(t)\|_{\hat{P}} = y(t) = e^{-\alpha_N t} v(t) \leq e^{-\alpha_N t} s(t)$, which yields \eqref{eq: X bound degenerate}. The initial conditions $s(0) = y(0)$ and $\dot s(0) = f\big(0, s(0) \big)$ lead to
    \begin{equation*}
        h_+ + h_- = \|\chi(0)\|_{\hat{P}} \quad \text{and} \quad p + h_+ \alpha_N - h_- \alpha = (\alpha_N - \alpha) \|\chi(0)\|_{\hat{P}} - b_{min}^{\hat{P}} + \gamma \|x_N(0)\|_{P_N}.
    \end{equation*}
    We can solve these equations as
    \begin{equation*}
        h_{\pm} = \frac{\frac{1}{2}\big(-\alpha_N - \alpha \pm 3(\alpha - \alpha_N) \big)\|\chi(0)\|_{\hat{P}} \mp \gamma \|x_N(0)\|_{P_N} \pm b_{min}^{\hat{P}} \pm p}{\alpha_N + \alpha}.
    \end{equation*}
\end{proof}

\begin{spacing}{0.8}
\bibliographystyle{IEEEtran}
\bibliography{references}
\end{spacing}

\end{document}